\documentclass{article}
\usepackage{graphicx} 
\usepackage{subcaption}
\usepackage{url}
\usepackage{xargs} 
\usepackage{comment}
\usepackage[linesnumbered,ruled,vlined]{algorithm2e}
\usepackage[colorinlistoftodos,prependcaption,textsize=tiny]{todonotes}
\usepackage{amssymb}
\usepackage{amsmath}
\usepackage{authblk}
\usepackage{hyperref}
\usepackage{amsthm}
\usepackage{soul}
\usepackage{float}
\usepackage{pgfplots}
\usepackage{booktabs,tabularx}
\pgfplotsset{compat=1.18} 

\newtheorem{theorem}{Theorem}
\newtheorem{lemma}{Lemma}
\newtheorem{definition}{Definition}

\title{SonicDB S6: A Storage-Efficient Verkle Trie for High-Throughput Blockchains}

\date{\vspace{-5ex}}

\begin{document}

\author[]{Luigi Crisci}
\author[]{Lorenz Sch\"uler}
\author[]{Herbert Jordan}
\author[]{Bernhard Scholz}
\affil[]{Sonic Labs}

\maketitle
\begin{abstract}
    The Ethereum state database uses  Merkle Patricia Trie (MPT), which suffers from large witness proof sizes and high storage overhead.
    Verkle Tries have been proposed as a replacement, offering witness proofs below 150 bytes through vector commitments and Inner Product Argument aggregation. However, deploying a Verkle Trie in a high-throughput, short block-time blockchain such as Sonic, which produces a block every 300 milliseconds, introduces substantial engineering challenges related to storage efficiency, commitment computation costs, and the need to serve both live and historical state queries in real time.

    We present SonicDB S6, a production-grade Rust Verkle Trie database for the Sonic blockchain, which leverages its non-forking property to enable aggressive storage optimizations. Occupancy-aware node specializations, selected via an $\mathcal{O}(k n^2)$ dynamic program, reduce live storage by 97.8\%. Delta nodes that record only changed slots reduce archive storage by 95\%. Batched updates, multi-threaded commitment computation, and homomorphic Pedersen caching yield $3.2\times$ higher throughput than a persistent Geth Verkle baseline while sustaining production block-rate performance.
\end{abstract}

\section{Introduction}

The Ethereum state layer stores balances, nonces, byte-code, and storage slots of accounts and is implemented as a Merkle Patricia Trie (MPT)~\cite{merkle1988digital,morrison1968patricia,wood2014ethereum} that
suffers several limitations.
Verkle Tries~\cite{verkle-eip,VTMPTbenchmarking,buterin2021verkle} have been proposed as a drop-in replacement for the Merkle Patricia Trie (MPT)~\cite{merkle1988digital,morrison1968patricia,wood2014ethereum}.
Verkle Tries replace cryptographic hashing at each node of the Trie with vector commitments~\cite{vector-commitments}, thereby allowing commitment to large amounts of data with a single group element.
The key benefit is that they support \textit{position-binding openings}~\cite{vector-commitments}, resulting in witness proofs~\cite{storage-proof} of an order pair without supplying any sibling data.
Witness proofs in the Merkle Patricia Trie require supplying all sibling nodes along the root-to-leaf path, resulting in witness proofs of on the order of several kilobytes per key-value pair.
At block scale, the aggregate witness~\cite{ethereum-stateless} for a single Ethereum block can reach tens of megabytes, making \textit{stateless clients}~\cite{ethereum-stateless} impractical.
Network nodes that verify blocks without retaining the full state cannot verify them as efficiently when witness data is this large.
In addition, the trie's hexadecimal branching factor and variable-length encoding lead to deep, unbalanced structures with high read and write amplification~\cite{readamp}, placing mounting pressure on storage as state grows.

Ethereum Verkle Tries overcome the shortcomings of MPT's witness proofs using Inner Product Argument (IPA) proof aggregation~\cite{bootle2016efficient,bulletproofs,halo}, which reduces proof witness sizes from several kilobytes to under 150 bytes regardless of the Trie's depth, enabling technologies such as stateless clients in the future.
The higher branching factor of 256 also limits tree depth to at most 32 levels for 32-byte keys and unifies account and contract storage into a single key space, eliminating the per-contract storage tries of the MPT design. The advantages of Verkle's proof witnesses have been experimentally demonstrated in~\cite {VTMPTbenchmarking}.

Despite their theoretical appeal, building a Verkle Trie database suitable for production use in a high-throughput blockchain with very short time-to-finality, such as Sonic~\cite{soniclabs}, introduces substantial engineering challenges.
For example, SonicDB is developed for the Sonic blockchain, which produces a block every 300 milliseconds — 40 times faster than Ethereum's 12-second block interval.
Every commitment recomputation, node write, and storage allocation must complete within this tight budget.

In addition, Verkle nodes nominally store 256 values or child pointers each. While the high branching factor keeps the Trie shallow, it also means that nodes are frequently sparsely populated, particularly at lower levels.
A leaf node storing a single 32-byte value would still occupy more than 8~KiB if stored at full size.
Storing all 256 slots regardless of occupancy results in enormous waste; designing a storage structure that adapts to actual usage while keeping lookup overhead low is therefore paramount.
Another key challenge in any Verkle Trie implementation is the cost of recomputing the state root commitment after each block.
A Pedersen commitment requires elliptic-curve scalar multiplication, which is roughly two to three orders of magnitude slower than the hash evaluations used in MPTs.
Incremental updates mitigate this cost during ordinary state transitions, but the state root must be finalized within the block production window, leaving little room for unoptimized computation.

Our Verkle Trie uses two distinct databases serving fundamentally different roles. The \textit{LiveDB} maintains the current chain state for validators and observers, requiring low-latency reads and writes while retaining no history. The \textit{ArchiveDB} supports historical queries by preserving the full sequence of past states and requires non-destructive, copy-on-write updates. Both databases must remain synchronized in real time: if the ArchiveDB falls behind the LiveDB, it becomes a bottleneck for the entire system. Serving both roles efficiently with a single underlying Verkle trie design requires careful architectural separation.

The Ethereum committee phased out the development of Verkle Tries in accordance with the NIST recommendation to deprecate elliptic curve cryptography~\cite{ecc-phase-out}, and pivoted to a new hash-based state tree design~\cite{eip7864}.
However, the Sonic chain characteristics justify exploring Verkle Tries as a state database. In particular, the higher storage footprint of the Sonic chain makes zero-knowledge proof generation for stateless clients more appealing.

This paper presents SonicDB S6, a Rust implementation of an EIP-6800-compatible Verkle Trie state database for the Sonic blockchain.
SonicDB exploits the fact that Sonic is a non-forking blockchain, which eliminates the need to maintain diverging versions of the database at the same block height and enables more aggressive storage optimizations than would be possible in a general-purpose setting. SonicDB S6 makes the following contributions:
\begin{itemize}
    \item We introduce occupancy-aware node specializations with a subsumption property and optimal selection via dynamic programming, reducing the LiveDB storage by 97.8\% (from 817.1 GiB to 17.9 GiB) across 40 M blocks.
    \item  We propose delta nodes that store only changed slots relative to a base node while avoiding chaining, which together with specialization reduces storage of the ArchiveDB by 95\% (from 16,141 GiB to 809 GiB).
    \item We present the SonicDB S6 architecture with batched updates, multi-threaded commitment computation, and cached Pedersen commitment state, achieving 2.85× higher throughput than a persistent Geth Verkle baseline while keeping archive performance at production block rate.
\end{itemize}
The remainder of this paper is organized as follows.
Section~\ref{sec:background} provides background on Ethereum Verkle Tries, their structure, and their proof-size advantages over MPTs.
Section~\ref{sec:live-archive} describes our forkless live and archive database design, including delta nodes.
Section~\ref{sec:specialization} presents the node specialization framework and its dynamic programming solution. Section~\ref{sec:implementation} details the implementation layers and commitment computation optimizations.
Section~\ref{sec:evaluation} evaluates storage efficiency and throughput against the Geth baseline.
Section~\ref{sec:related} surveys related work, and Section~\ref{sec:conclusion} concludes with directions for future work.

\section{Background: Verkle Tries \label{sec:background}}

In the following section, we provide an overview of important concepts related to Verkle Tries, following the Ethereum specification \cite{verkle-eip}.

\paragraph{Verkle Commitments.}
A \textit{commitment scheme} is a cryptographic primitive that allows one party (the committer) to bind themselves to a chosen value without revealing it to a second party (the verifier).
\textit{Vector commitments}~\cite{vector-commitments} extend this primitive to sequences of values $(v_1, v_2, \ldots, v_n)$ rather than a single message.
In addition to the standard binding and hiding properties, vector commitments support \textit{position-binding openings}: the committer can prove that $v_i$ occupies position $i$ in the committed sequence without disclosing the remaining values.
This is the crucial property that enables compact proofs in Verkle tries, since a prover can certify a single key-value pair by producing an opening for that leaf's position in the commitment, rather than hashing all sibling nodes as in a Merkle proof.
Vector commitments also offer an \textit{updatability} property: if $v_i$ is replaced by $v_i'$, the commitment can be updated incrementally without recomputing the cryptographic hashes from scratch.

Ethereum's Verkle trie proposal uses the \textit{Pedersen commitment
    scheme}~\cite{pedersen1992non}, combined with the \textit{Inner Product Argument}
(IPA)~\cite{bootle2016efficient,bulletproofs} for efficient position-binding
openings~\cite{vector-commitments}.
The Pedersen scheme is based on elliptic curve cryptography and derives its security from the \textit{discrete logarithm assumption}: given a prime-order group $\mathbb{G}$ of order $p$ and two independently chosen generators $g_1, g_2 \in \mathbb{G}$ whose discrete logarithm relationship is unknown, it is computationally infeasible to find $x$ such that $g_1 = x \cdot g_2$.
Under this assumption, a Pedersen commitment to a pair of field elements $v_1, v_2 \in \mathbb{F}_p$ is $C = v_1 g_1 + v_2 g_2$, which is a single group element. The scheme generalizes to arbitrary-length vectors: given a set of generators $g_1, \ldots, g_n$ with no known discrete logarithm relationships, a commitment to $(v_1, \ldots, v_n)$ is $C = \sum_{i=1}^{n} v_i \, g_i$.
A key algebraic property of Pedersen commitments is \textit{additive homomorphism}: for two value vectors $V_0$ and $V_1$ with respective commitments $C_0$ and $C_1$,
\begin{equation*}
    C_0 + C_1
    = \operatorname{Commit}(V_0) + \operatorname{Commit}(V_1)
    = \operatorname{Commit}(V_0 + V_1).
\end{equation*}
This property is what makes IPA-based openings efficient: rather than revealing the full committed vector, a prover can convince a verifier of an inner product relation through a logarithmic-round interactive protocol, which translates into a compact non-interactive proof via the Fiat--Shamir heuristic~\cite{fiat1986prove}.

In the Ethereum proposal, Pedersen commitments are instantiated over the \textit{Bandersnatch} elliptic curve~\cite{masson2024bandersnatch}.
Bandersnatch is a twisted Edwards curve defined over the BLS12-381 scalar field, chosen because its native field arithmetic aligns with BLS12-381-based SNARK circuits, making proof verification highly efficient inside zero-knowledge proofs.
Importantly, the Bandersnatch group order $p$ is a 253-bit prime, which means the curve can natively commit to values of at most 252 bits.
Values of 256 bits (i.e.\ 32 bytes), as used throughout the Ethereum state, therefore exceed the group order and must be handled with care, as described in the extension node encoding below.

It is worth noting that computing a Pedersen commitment is considerably more expensive than computing a cryptographic hash.
A commitment requires $n$ scalar multiplications on an elliptic curve, each of which is roughly two to three orders of magnitude slower than a SHA-256 evaluation.
Incremental updates mitigate this cost during ordinary state transitions, but initial Trie construction and proof generation remain computationally expensive.

\paragraph{Node Structure.}
\begin{figure}[htbp]
    \centering
    \includegraphics[width=0.8\textwidth]{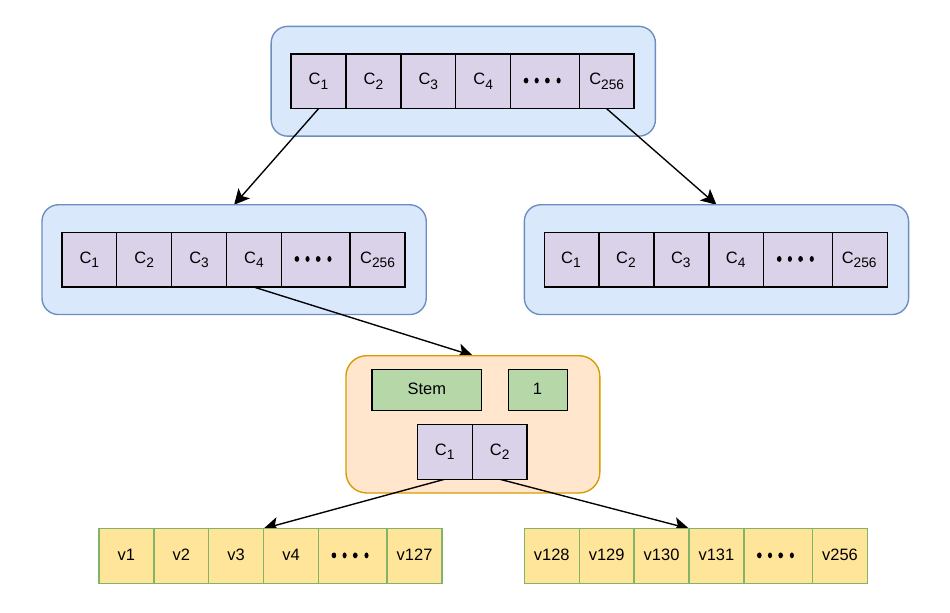}
    \caption{Ethereum Verkle Trie structure.}
    \label{fig:verkle-trie-composition}
\end{figure}
The Ethereum Verkle Trie is composed of two types of nodes: \textit{inner} nodes and \textit{extension} nodes, with a branching factor (arity) of 256.
Figure~\ref{fig:verkle-trie-composition} illustrates an example instance of a Verkle Trie.
Inner nodes (shown in blue in Figure~\ref{fig:verkle-trie-composition}) act as the internal branching structure of the trie.
Each inner node has up to 256 children, which may themselves be either inner nodes or extension nodes.
The commitment of an inner node is the Pedersen commitment of its children's commitments:
\begin{equation}
    C_{\mathrm{inner}} = \operatorname{Commit}(C_1, C_2, \ldots, C_{256}),
\end{equation}
where $C_i$ denotes the commitment of the child at position $i$.
This recursive structure allows any sub-tree to be summarized by a single group element, enabling constant-size proofs regardless of Trie depth.
Extension nodes (shown in orange in Figure~\ref{fig:verkle-trie-composition}) form the leaves of the Trie and store up to 256 key-value pairs, indexed by the final byte of the 32-byte key.
The first 31 bytes of the key, shared across all values in the node, are explicitly stored as the node's \textit{stem}, permitting early termination of lookups when no matching prefix exists.

Because Bandersnatch has a 253-bit group order---smaller than the 256-bit range of Ethereum state values---extension nodes cannot commit to 32-byte values directly~\cite{verkle-eip}.
To resolve this, each value $v_i \in \mathbb{B}_{32}$ is split into two 16-byte halves: a lower half $v_i^{l} \in \mathbb{B}_{16}$ and an upper half $v_i^{h} \in \mathbb{B}_{16}$.
Both halves fit comfortably within the curve's 252-bit capacity.
To distinguish between a leaf that does not exist and one that is explicitly set to zero, a marker bit is embedded in the lower half: the 129th bit of $v_i^{l}$ is set to 1 if the value is present and 0 if absent, yielding the modified lower half $v_i^{l,\mathrm{mod}}$.
This encoding is required to prevent an adversary from deleting state by writing zero values.
The 256 modified value pairs are then committed as two intermediate Pedersen commitments, each covering half the leaf range:

\begin{align}
    C_1 & = \operatorname{Commit}
    \bigl(v_{0}^{l,\mathrm{mod}},\, v_{0}^{h},\,
    v_{1}^{l,\mathrm{mod}},\, v_{1}^{h},\,\ldots,\,
    v_{127}^{l,\mathrm{mod}},\, v_{127}^{h}\bigr), \\
    C_2 & = \operatorname{Commit}
    \bigl(v_{128}^{l,\mathrm{mod}},\, v_{128}^{h},\,
    v_{129}^{l,\mathrm{mod}},\, v_{129}^{h},\,\ldots,\,
    v_{255}^{l,\mathrm{mod}},\, v_{255}^{h}\bigr).
\end{align}

Finally, the commitment of the extension node as a whole incorporates a version marker (currently set to 1 to facilitate future state-expiry schemes~\cite{ethereum-state-expiry}), the stem, and both sub-commitments:
\begin{equation}
    \label{eq:leaf-commitment}
    C_{\mathrm{ext}} = \operatorname{Commit}(1,\; \mathrm{Stem},\; C_1,\; C_2).
\end{equation}

\paragraph{Witness Proof}
A \textit{witness proof} (or \textit{membership proof}) is a cryptographic proof that a given key-value pair $(k, v)$ is included in the Trie, without requiring the verifier to hold the entire Trie structure.
Witness proofs are essential for \textit{stateless clients}~\cite{ethereum-stateless}: nodes that verify blocks without storing the full blockchain state.

In an MPT, the integrity of each node is guaranteed by its parent storing the cryptographic hash of that node's contents.
To prove that a key-value pair $(k, v)$ is present, a prover must supply a \textit{Merkle proof}: the sequence of all sibling nodes encountered on the root-to-leaf path.
The verifier recomputes each node's hash from the leaf upward, checking at every level that the computed hash matches the value recorded in the parent, until the root hash is reconstructed and compared against a trusted value.
At each level, the branching factor $b$ determines the number of siblings, each contributing one hash (32 bytes for Keccak-256) to the proof.
With $b = 16$ and depth $d \approx 8$--$9$ (since $16^8 \approx 4 \times 10^9$), each level contributes up to $(b-1) \times 32 = 15 \times 32 = 480$ bytes, yielding witness sizes on the order of several kilobytes per key-value pair.

Verkle Tries replace hashing with \textit{vector commitments}, which support position-binding openings.
A prover can reveal a single position $v_i$ without exposing other values, using an IPA proof~\cite{bootle2016efficient,bulletproofs} of size logarithmic in the branching factor.
Moreover, IPA proofs across levels can be \textit{aggregated}~\cite{halo}, so a full root-to-leaf witness fits in roughly 100--150 bytes regardless of Trie size.

Two effects compound this improvement: the higher branching factor halves tree depth (since $\log_{256} N = \tfrac{1}{2}\log_{16} N$), and position-binding openings contribute a \textit{constant} per level rather than $\mathcal{O}(b)$ hashes.
Together, they reduce witness sizes by roughly two orders of magnitude.

\begin{table}[ht!]
    \centering
    \begin{tabular}{lcc}
        \toprule
        \textbf{Property}       & \textbf{MPT}   & \textbf{Verkle Trie}     \\
        \midrule
        Branching factor $b$    & 16             & 256                      \\
        Proof element per level & $(b-1)$ hashes & 1 IPA proof (aggregated) \\
        Bytes per level         & $\sim$480      & amortized $\sim$0--48    \\
        Depth for $10^9$ leaves & $\sim$8        & $\sim$4                  \\
        Witness size            & $\sim$3--4 KiB & $<$150 bytes             \\
        \bottomrule
    \end{tabular}
    \caption{Witness size comparison between Merkle Patricia Trie (MPT)~\cite{merkle1988digital,morrison1968patricia,wood2014ethereum}  and Verkle Trie~\cite{buterin2021verkle} for a trie
        holding approximately $10^9$ key-value pairs~\cite{VTMPTbenchmarking}.}
    \label{tab:witness-size}
\end{table}

\section{Dual-Mode Storage: LiveDB and ArchiveDB%
  \label{sec:live-archive}}

Blockchain state management imposes two fundamentally different access requirements on the underlying database.
Validators and observers require low-latency reads and writes against the \emph{current} world state: given the state at block $h$, they must produce the state at block $h+1$ with minimal latency, and do not need to retain any prior state of block height less than $h$.
Archive nodes, by contrast, must respond to historical queries at arbitrary block heights, which demands that every state transition be preserved non-destructively and remain efficiently retrievable long after it occurred.

In most blockchain implementations, both roles are served by a single general-purpose database, which forces an architectural compromise.
Retaining the full history of state transitions introduces storage and write overhead that impairs validator performance, while the low-latency write patterns suited to block production are ill-adapted to the append-only workloads characteristic of historical access.

Another complicating aspect of general-purpose designs is the need to support chain reorganizations, which require the database to maintain divergent state versions at the same block height until a single branch emerges as the chain's canonical branch.
Sonic, in common with a class of modern consensus protocols, provides a \emph{non-forking} guarantee: each block has at most one canonical successor, so divergent versions never arise.
This property removes the reorganization requirement entirely, and with it the need for the version-management and pruning machinery that general-purpose designs must carry.

We exploit this guarantee by decomposing the state layer into two specialized databases, LiveDB and ArchiveDB~\cite{11275644,jordan2025fastethereumcompatibleforklessdatabase}, each designed exclusively for its access pattern and each capable of applying optimizations that would be unsafe or inapplicable in a forking setting.

\begin{figure}[htbp]
    \centering
    \includegraphics[width=\textwidth]{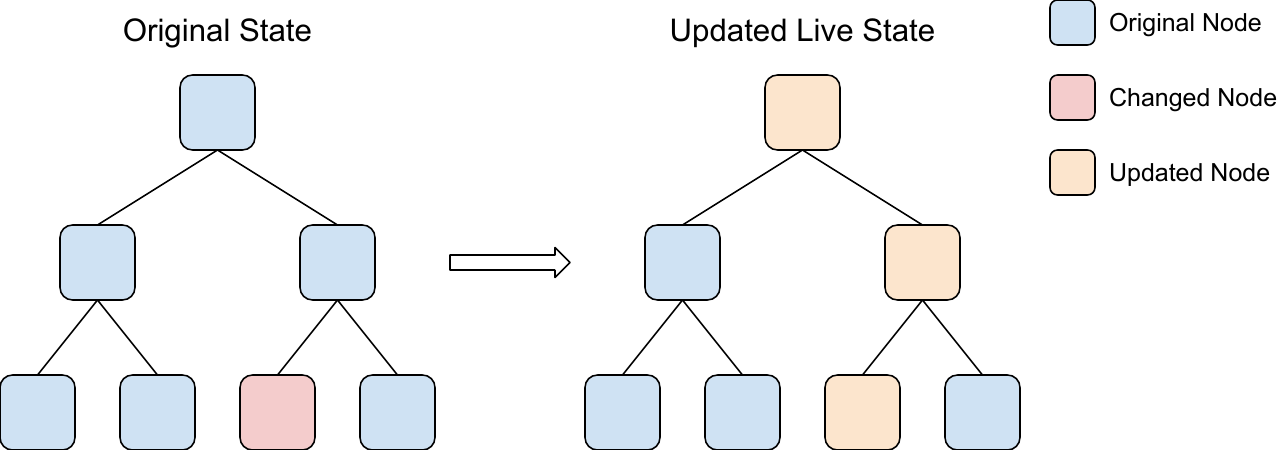}
    \caption{LiveDB update of a single leaf node. Because no historical data
        is retained, changes are applied in place.}
    \label{fig:live-update}
\end{figure}

\paragraph{LiveDB.} The LiveDB serves the validator and observer roles.
It maintains a single mutable image of the current world state and advances it from block $h$ to block $h+1$ via destructive in-place updates: modified values overwrite their predecessors directly, and no copy of the superseded state is retained.
Consequently, storage reclamation occurs intrinsically as part of the update path, with no separate compaction or pruning pass required.

Figure~\ref{fig:live-update} illustrates a leaf-node update in the LiveDB.
The absence of any versioning constraint means that nodes may be freely transformed between structural variants or deleted as block processing demands, without regard for any reader that might require an earlier version of the same node.
This freedom enables occupancy-aware node specialization described in the next section.
Every node can be held in whichever specialization minimizes its current storage footprint, and re-specialized whenever its occupancy changes, because the LiveDB carries no historical obligation.

\begin{figure}[htbp]
    \centering
    \includegraphics[width=\textwidth]{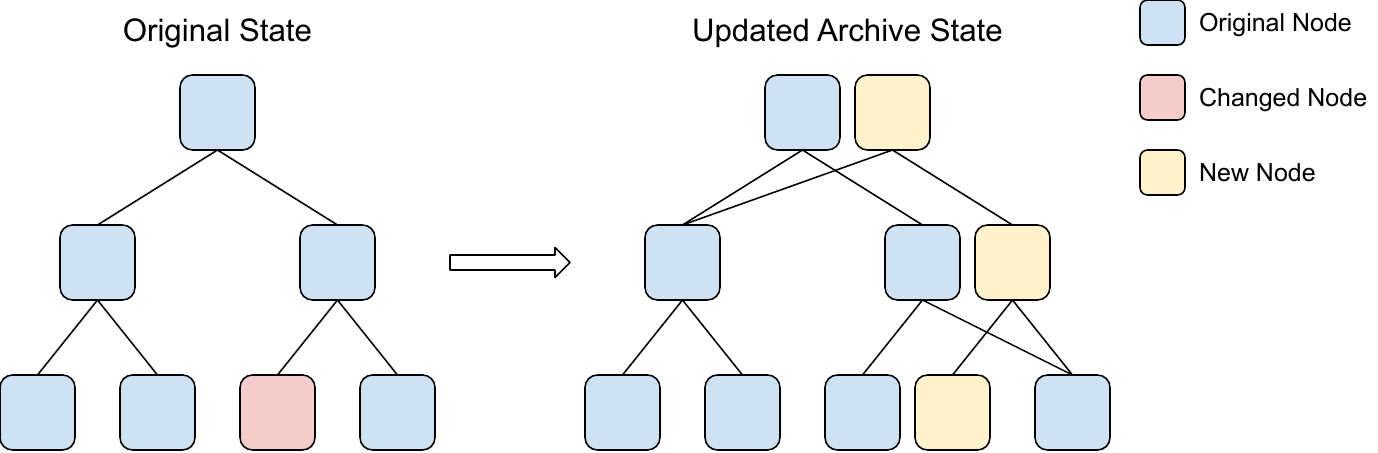}
    \caption{ArchiveDB update of a single leaf node. The changed node is
        copied and updated, and every ancestor up to the root is
        duplicated and updated with its new child identifier.}
    \label{fig:archive-update}
\end{figure}

\paragraph{ArchiveDB.} The ArchiveDB serves the archive node role.
It maintains a complete and queryable record of the world state at every block height through an append-only, copy-on-write update mechanism.
When a node is modified, a new clone of the node is produced with a fresh identifier, and the modification propagates up the path to the root, where each ancestor is likewise copied and updated with the new child identifier.
Nodes that are not touched by a given block are shared across versions without duplication, so the marginal storage cost of a block is proportional to the number of nodes on the modified root-to-leaf paths rather than to the size of the full Trie.

Figure~\ref{fig:archive-update} illustrates this process for a single leaf update.
The non-forking guarantee is essential here.
The ArchiveDB never forks the version graph, and the copy-on-write chain forms a simple linear sequence, because each block has at most one successor.
This eliminates the version-branching and garbage-collection complexity that would otherwise be required, and allows the archive to be structured as a compact, append-only log rather than a general persistent data structure.

The principal cost of the copy-on-write scheme is storage amplification~\cite{readamp}.
A single-slot update to a leaf node requires copying every node on the root-to-leaf path, and upper-level inner nodes, which are typically dense and are touched by nearly every block, are therefore duplicated at high frequency even though only a small fraction of their 256 child identifiers change per update.
We address this amplification directly through the delta node mechanism.

Delta nodes reduce the storage overhead of the ArchiveDB by recording only the slots that changed relative to a recent base node, rather than duplicating the full 256-slot layout on every copy-on-write operation.
The idea of associating a small set of modifications with a base representation, and promoting to a fresh base once the modification budget is exceeded, has precedent in the fat-node-with-modification-box technique of Driscoll et~al.~\cite{driscoll1989persistent} and in the delta-record-plus-consolidation design of the Bw-tree~\cite{levandoski2013bwtree}.
Under copy-on-write semantics, every block update to a leaf node forces a fresh copy of all inner nodes on the root-to-leaf path.
Upper-level inner nodes are therefore rewritten on almost every block, yet typically only a small number of their 256 child identifiers change.
Without delta compression, the storage cost per block update to a single slot is $\Theta(d \cdot W)$, where $d$ is the depth of the modified node and $W = 256$ is the node width.
For a fully populated inner node near the root, $W \cdot 32$ bytes of child identifiers are copied even though all but a constant number are unchanged.

Let $N_b$ be a \emph{base node} at block height $h_b$, storing the full array $V_b \in \mathbb{B}_{32}^{256}$ of child identifiers (for an inner node) or values (for a leaf node).
A \emph{delta node} $N_\delta$ at block height $h_\delta > h_b$ is a pair
\[
    N_\delta = (\mathit{id}(N_b),\; \Delta),
\]
where $\Delta = \{(i_1, v_1), \ldots, (i_m, v_m)\}$ is the set of $m$ changed slots, with $i_j \in [0, 255]$ distinct and $v_j \in \mathbb{B}_{32}$. The logical content of $N_\delta$ is the array $V_\delta$ defined by
\[
    V_\delta[i] =
    \begin{cases}
        v_j    & \text{if } (i, v_j) \in \Delta, \\
        V_b[i] & \text{otherwise.}
    \end{cases}
\]
A \emph{delta threshold} $\tau \in [1, 256]$ governs when a delta node is promoted to a base node: if applying a further update would cause $|\Delta|$ to exceed $\tau$, a new base node is written instead, and subsequent deltas reference it.
Experiments have shown that $\tau = 10$ produces good compression ratios.
Figure \ref{fig:delta-node-conversion} shows the conversion between a base node and a delta node.

Every delta node references a base node directly and never references other delta nodes.
This is enforced by the promotion rule: once $|\Delta|$ would exceed $\tau$, a full base node is materialized.
The invariant guarantees that reading the logical content of any historical node requires exactly two
node loads (one for the delta and one for its base), giving a read amplification factor of two, independent of block height or update frequency.
Without this invariant, a chain of $c$ deltas would require $c+1$ loads, making the read cost unbounded as the chain's length increases.

Given an existing delta node $N_\delta = (\mathit{id}(N_b), \Delta)$ and an incoming update set $U = \{(i, v)\}$, define the merged change set $\Delta' = (\Delta \setminus \{(i, \cdot) : i \in \pi_1(U)\}) \cup U$, where $\pi_1$ projects onto slot indices.
If $|\Delta'| \leq \tau$, a new delta node $(\mathit{id}(N_b), \Delta')$ is written. If $|\Delta'| > \tau$, a new base node is materialized from $V_\delta$ with $U$ applied, and future delta nodes will reference it.
Critically, the new delta always references $N_b$, not the previous $N_\delta$: this replaces the old delta entry in storage rather than appending to a chain.

The on-disk representation of a delta node occupies
\[
    S_\delta(m) = S_{\mathit{id}} + m \cdot (1 + 32) \;\text{bytes},
\]
where $S_{\mathit{id}}$ is the size of a node identifier (6 bytes in our
implementation), $1$ byte encodes the slot index $i \in [0, 255]$, and $32$
bytes hold the value $v_j$. For $m \leq \tau \ll 256$, this is substantially
smaller than the $256 \cdot 32 = 8{,}192$ bytes of a full node and even smaller than large specializations.
A base node occupies the full specialization size as determined by the occupancy-aware
layout described in Section~\ref{sec:specialization}.
Across a sequence of $B$ blocks in which a given inner node changes $k$ slots per block on average, the total storage cost with delta nodes is in the worst case
\[
    \left\lfloor \frac{B \cdot k}{\tau} \right\rfloor \cdot S_{\mathit{base}}
    \;+\;
    B \cdot S_\delta(k),
\]
compared with $B \cdot S_{\mathit{base}}$ without delta nodes, giving a
compression ratio approaching $S_{\mathit{base}} / S_\delta(k)$ when
$k \ll \tau$.
In practice, the compression ratio is even higher because if the same key is updated multiple times, no additional slot in the delta is needed for subsequent updates.

The cached representation of a delta node extends its on-disk form with a copy of the slot values of $V_b$.
This avoids a secondary base-node lookup on every slot access if the slot is not in the delta once the node is resident in the cache.
When a delta node is loaded from disk, the node manager fetches $N_b$ (which may itself be served from cache), copies its data into the delta.
When a write guard is dropped, only the on-disk form  (identifier plus change set) is flushed, keeping write amplification proportional to $|\Delta|$ rather than to the full node width.
This is illustrated in Figure \ref{fig:delta-node-loading}.

\begin{figure}[htbp]
    \centering
    \includegraphics[width=0.8\textwidth]{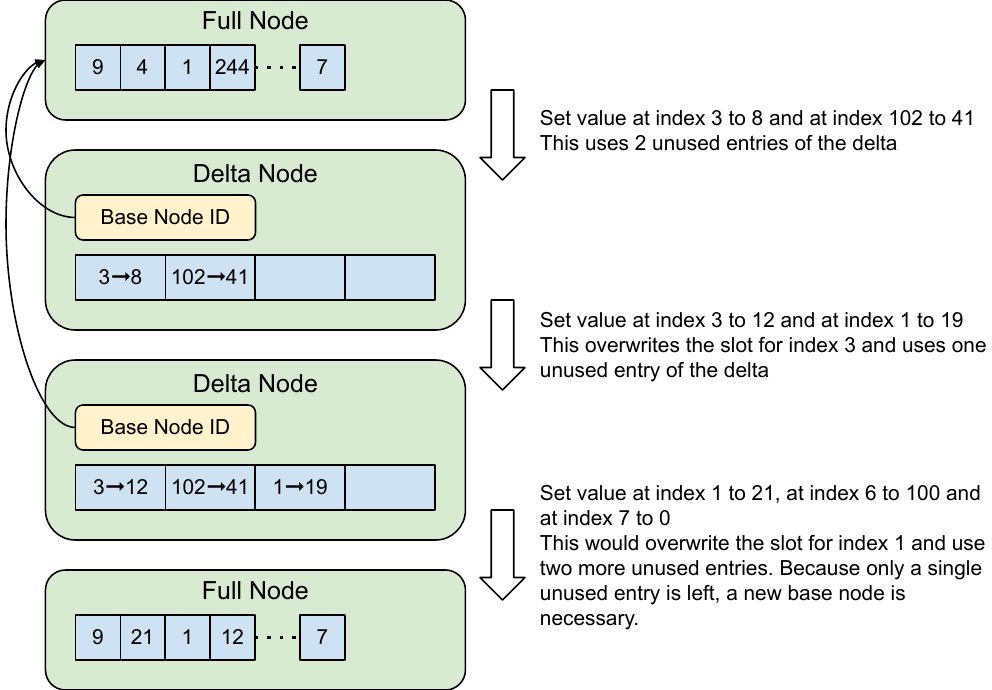}
    \caption{Conversion between a base node (e.g., a full node) and a delta node with delta size~$\tau = 4$.}
    \label{fig:delta-node-conversion}
\end{figure}

\begin{figure}[htbp]
    \centering
    \includegraphics[width=\textwidth]{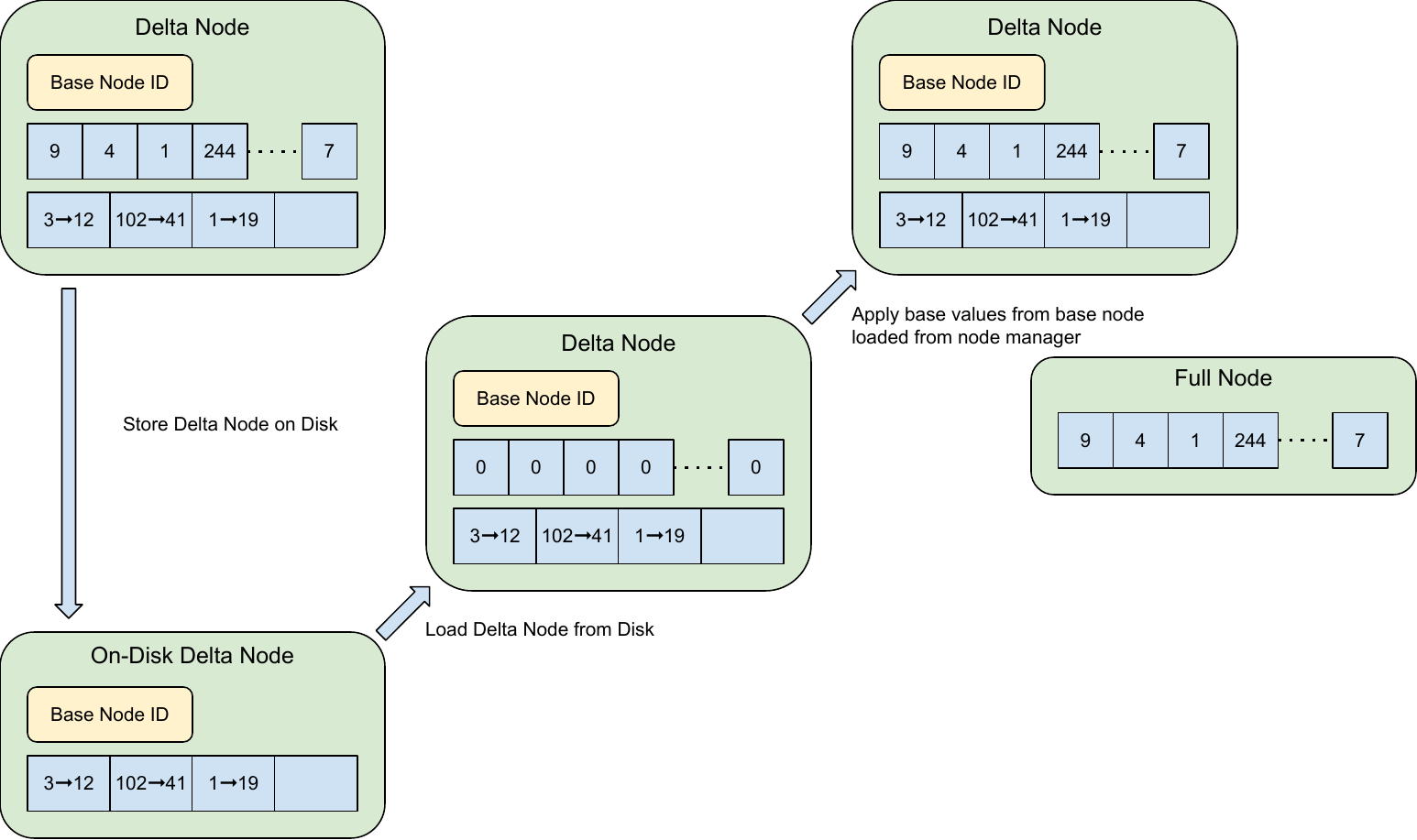}
    \caption{Loading a delta node: the on-disk form is read first, then the
        base node (e.g., a full node) is fetched and merged to produce the materialized view.}
    \label{fig:delta-node-loading}
\end{figure}

\section{Occupancy-aware Node Specialization \label{sec:specialization}}

Verkle nodes store 256 children or values, depending on whether the node is a leaf.
Because the Trie is sparse (see the experimental section), in practice, many nodes do not use all their available slots in the node data structure.
For example, leaf nodes at lower levels of the Trie often store only a few values, whereas inner nodes at lower levels often store only a few child IDs.
Storing all 256 slots per node, especially when the majority are empty, would consume substantial space due to the sheer number of nodes and leaves in the Verkle Trie.

\begin{figure}[htbp]
    \centering
    \begin{subfigure}{0.49\textwidth}
        \centering
        \includegraphics[width=\textwidth]{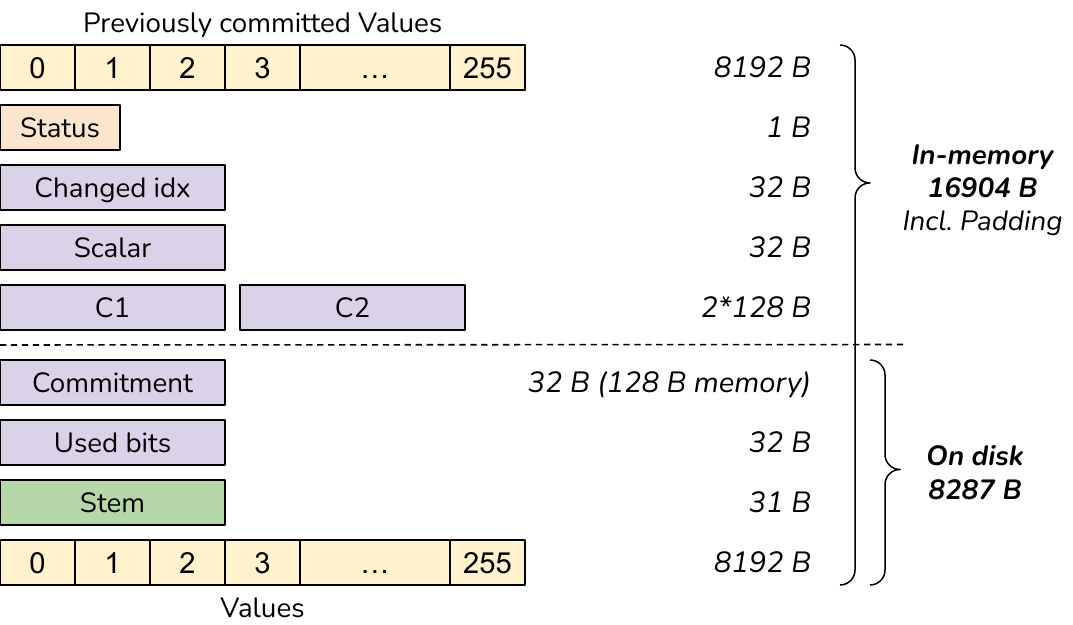}
        \caption{Full node}
        \label{fig:node-structure-full-leaf}
    \end{subfigure}
    \begin{subfigure}{0.49\textwidth}
        \centering
        \includegraphics[width=\textwidth]{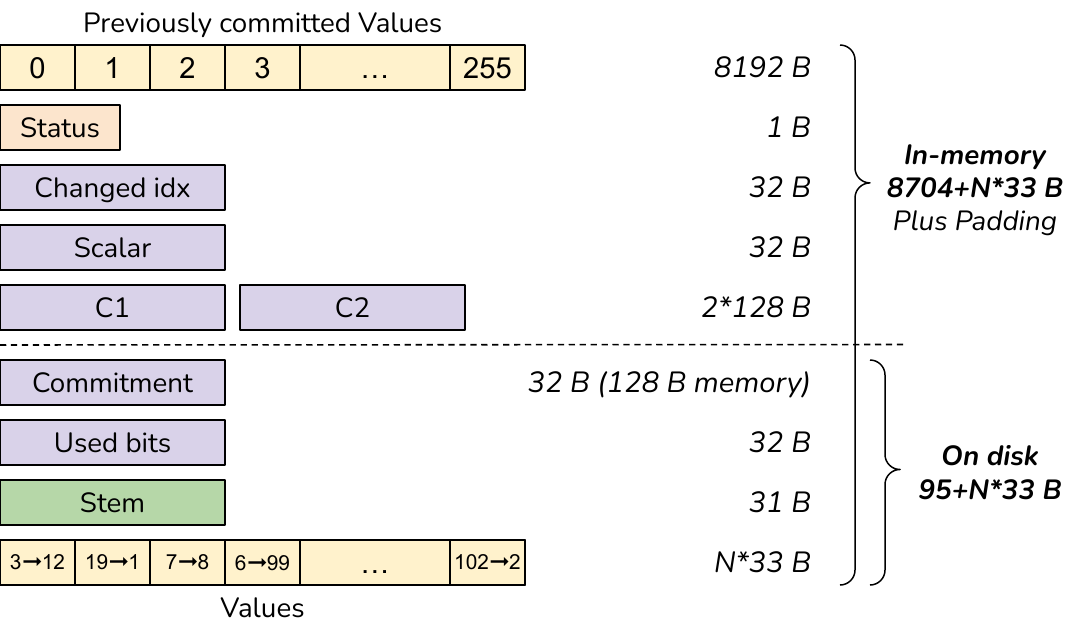}
        \caption{Sparse node}
        \label{fig:node-structure-sparse-leaf}
    \end{subfigure}\\[1em]
    \begin{subfigure}{0.49\textwidth}
        \centering
        \includegraphics[width=\textwidth]{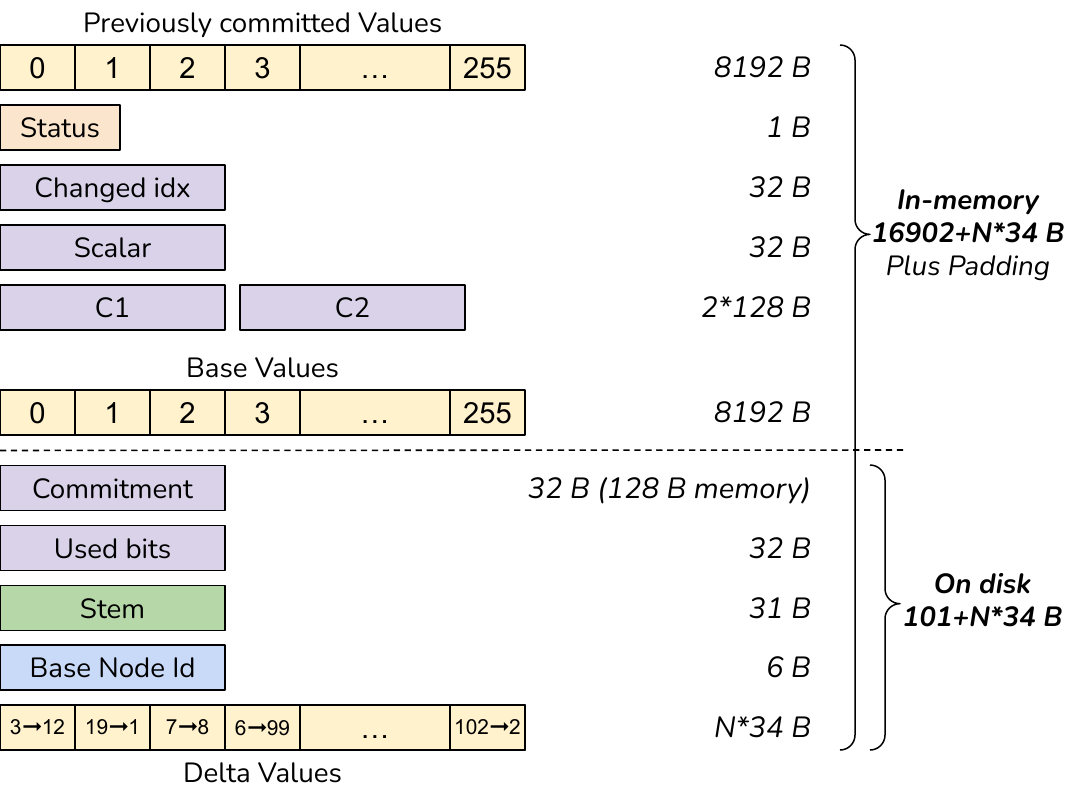}
        \caption{Delta node}
        \label{fig:node-structure-delta-leaf}
    \end{subfigure}
    \caption{Composition details of Verkle leaf node variants}
    \label{fig:node-structure}
\end{figure}

Figure \ref{fig:node-structure} shows the structure of our leaf nodes, which is similar for inner nodes.
In the bottom half, the 256 value slots, the leaf stem, the used bits, and the commitment are stored on disk.
The Verkle commitment is stored in a compressed 32-byte format and decompressed upon loading into main memory.
The top half, instead, contains only cached data used for commitment computation, making up $\sim 50 \%$ of the node size.

One could argue that finding the optimal number of specializations to minimize space usage is straightforward: one specialization per node type and the number of possible children/values, e.g., a leaf node with one child, an inner node with one value, a leaf node with two children, etc.
We refer to the specialization of an inner node as $\texttt{Inner}_i$, where $i$ is the maximal number of children of the specialization, and of a leaf node as $\texttt{Leaf}_i$, where $i$ is the maximal number of values.
However, having one specialization per node type requires that, whenever a new child or value is added to an inner or leaf node, the node be specialized to include one more child or value.

For example, a leaf node with three values is transformed to a leaf node with four values $\texttt{Leaf}_3 \rightarrow \texttt{Leaf}_4$.
Transforming a specialization to another specialization can be expensive because it requires copying all data to a new object, obtaining a new storage ID, and freeing the previous ID, which must traverse the entire storage stack.
Note that the individual operation is inexpensive in isolation, but it becomes expensive when many such transformations must be performed to maintain the database state when node specializations are used.

Another downside of having many specializations is that they do not integrate well with node deletions and the reuse of disk space for deleted nodes.
The offset of a node is encoded in its ID.
The ID, in turn, is stored in its parent node.
When a node is deleted, the nodes that follow it in the file are not shifted backward, as this would be very expensive.
Instead, the deleted node's offset is added to the free list, and the next time a new node is created, the offset is reused.
As the tree grows, the node distribution changes, and the number of nodes in a given specialization may increase or decrease.
This can temporarily result in many reusable offsets, which means disk space is occupied but not used to store actual data.
Having more specializations would lead to greater changes in the node distribution and more reusable nodes, resulting in more occupied but unused disk space.
Note, however, that having more specializations does not lead to write amplification, as the number of dirty nodes does not change.
It does not matter whether to update the existing node, which is marked as dirty, or a new node is created -- either way, there is one dirty node to be eventually written to disk when it gets evicted from the cache.

To avoid the copying, we have a \emph{subsumption property}.
A specialization of type $\textit{Leaf}_i$ can accommodate one to $i$ children; similarly, $\textit{Inner}_i$ can accommodate one to $i$ values.
For example, if the node specialization $\textit{Leaf}_4$ can be used for node types $\textit{Leaf}_4$, $\textit{Leaf}_3$, $\textit{Leaf}_2$, and $\textit{Leaf}_1$.
However, the fewer children/values a node type has in its specialization, the greater the space loss, since unoccupied leaves and children still occupy space in the specialization.

\paragraph{Specialization Functions.}
We model storage optimization as a mathematical problem by introducing a \emph{specialization function} that maps each node type to a representative specialization.
\begin{definition}\label{def:specialization}
    A \emph{specialization function} is a surjective function $\varphi\colon [1, n] \to X$ satisfying:
    \begin{itemize}
        \item $|X| = k$ \hfill\textbf{(cardinality constraint)}
        \item $X \subseteq [1, n]$ \hfill\textbf{(domain constraint)}
        \item $\forall\, i \in [1, n]\colon i \leq \varphi(i)$ \hfill\textbf{(coverage constraint)}
    \end{itemize}
\end{definition}
The function $\varphi$ maps each node type $i$ to a specialization $i' = \varphi(i)$. Surjectivity ensures that all $k$ specializations in $X$ are actually used.
The domain constraint requires every specialization to be a valid node type.
The coverage constraint $i \leq \varphi(i)$ ensures that a specialization has sufficient capacity for every node type assigned to it, since a node type with $i$ children must be covered by a specialization with at least $i$ children.
We denote the set of all specialization functions with parameters $n$ and $k$ by $\Phi^k_n$; every $\varphi \in \Phi^k_n$ must satisfy the three constraints of Definition~\ref{def:specialization}.
Note that $n \in X$ is forced: if $\varphi(n) \neq n$ for any $\varphi \in \Phi^k_n$, the coverage constraint $n \leq \varphi(n)$ is violated.
We may therefore write $X = X' \cup \{n\}$, where $X' \subseteq [1, n-1]$ and $|X'| = k - 1$.

\paragraph{Space Cost.}
We define the space function $s: [1,n] \rightarrow \mathbb{N}$, denoting the number of bytes for a specialization.
For example, $i \mapsto q$ gives the number of bytes $q$ for specialization $i$.
We further assume that the space function is a monotonic function with the property:
\[
    \forall i,j \in [1,n]: i \leq j \implies s(i) \leq s(j)
\]
With the coverage constraint, we can make the space function monotonic, as explained at the end of this sub-section.

We can further simplify the space function. We can express the costs $s(i) = h + j \delta$ where $h$ is a constant storage overhead per node specialization, and for each value/child, we have extra costs of $\delta$ bytes.
This, of course, depends on the system and programming language implementation, as well as on whether the specialization can be expressed as a linear function of the number of values/children.

\paragraph{Frequency Count.}
The mathematical model must account for the frequency of node types, as we observed in our experiments that the number of children/values per node is not uniformly distributed; some node types occur much less frequently than others, so introducing a specialization for them does not significantly reduce storage occupancy.
We denote the frequency of a node type with the following function $f: [1, n] \rightarrow \mathbb{N}$.

With functions $f$ for the frequency distribution and $s$ for space consumption, a mathematical program can be devised to express the optimization problem for node specialization.

\begin{definition}
    The cost of a specialization function $\varphi$ is defined as follows,
    \[
        \textit{cost}(\varphi) = \sum_{i \in [1,n]} s(\varphi(i)) f(i)
    \]
\end{definition}

The space consumption for a node type $i$ is the space consumption of its specialization $s(\varphi(i))$ multiplied by its frequency $f(i)$.
The total consumption is the sum of the consumption across all node types.

We note that $\textit{cost}(\varphi)$ models storage consumption only and does not include the throughput penalty that a larger specialization set imposes through branch prediction, cache pressure, and free-list fragmentation.
We treat $k$ as a separate design parameter chosen against a throughput constraint (Section~\ref{sec:evaluation}) and use the dynamic program below to solve the storage sub-problem exactly for a given $k$.
This decomposition is justified empirically: the storage function $\varphi$ is determined primarily by the observed frequency distribution $f$, whereas the throughput cost of an additional specialization is largely independent of where its threshold is placed.

The mathematical program searches for an optimal specialization function among many possible specialization functions:
\begin{definition}\label{def:mathprog}
    The node specialization problem is defined by finding the specialization function $\varphi$ with minimal cost $
        \min_{\varphi \in \Phi^k_n} \textit{cost}(\varphi)
    $.
\end{definition}

The mathematical program seeks a specialization function $\varphi$ with minimal space usage.
Among the optimal solutions $\arg \min_{\varphi \in \Phi^k_n} \textit{cost}(\varphi)$, there could be pathological solutions to the problem, in which $k$ is either $1$ or $n$.
If $ k=1$, we have a single specialization that must accommodate all possible occurring node types; hence, the solution is $i \mapsto n$ for all $i$, $1 \leq i \leq n$.
If $k=n$, we have $n$ specializations.
Each node type must become its own specialization, i.e., $i \mapsto i$.
Both solutions will be suboptimal in practice due to non-uniform node frequencies and specialization costs.
The sweet spot will be somewhere in between.

\paragraph{Normalized Specialization Functions.}

\begin{definition} \label{def:normalized-specialization}
    We define a normalized specialization function $\varphi_X^n$ by ordered set $X=\{x_1, \ldots, x_{k-1} \} \subseteq [1,n-1]$ with following construction for all $k$, $1<k\leq n$,
    \begin{align*}
        i & \mapsto x_1 & \mbox {if $i \in [1, x_1]$}                                   \\
        i & \mapsto x_l & \mbox {if $i \in [x_{l-1}+1, x_l]$, for all $l$, $1 < l < k$} \\
        i & \mapsto n   & \mbox {if $i \in [x_{k-1}+1, n]$}
    \end{align*}
    For $k=1$, the normalized specialization function is the mapping $i \mapsto n$, for all $i$, $1 \leq i \leq n$, where $X$ is the empty set.
\end{definition}
A normalized specialization is an $k$-partitioning of the $\varphi_X^n$'s domain into consecutive disjoint intervals from $1$ to $n$.

We define the set of all ranges for a normalized specialization function by $\mathbb{X}^{k-1}_{n-1}= \{ X \subseteq [1, n-1] \mid |X| = k-1 \}$.
The parameters of the normalized specialization function $\varphi_X^n$ is reduced by one since the range $X \in \mathbb{X}^{k-1}_{n-1}$ of a normalized function $\varphi_X^n$ is the set $X\cup\{n\}$.

\begin{lemma}\label{lem:norm-subset}
    \begin{equation}
        \forall X \in \mathbb{X}^{k-1}_{n-1}: \varphi_X^n \in \Phi^k_n
    \end{equation}
\end{lemma}
\begin{proof}
    We verify the three conditions of Definition~\ref{def:specialization}.

    \textbf{(Domain constraint)} By construction, $\varphi_X^n$ maps every $i \in [1,n]$ to some element of $X \cup \{n\} \subseteq [1,n]$.

    \textbf{(Cardinality constraint)} The range of $\varphi_X^n$ is exactly $X \cup \{n\}$, which has $|X| + 1 = (k-1)+1 = k$ elements, since $n \notin X$ as $X \subseteq [1,n-1]$.

    \textbf{(Surjectivity)} Each element of $X \cup \{n\}$ is the image of at least one element: every $x_l \in X$ is the image of itself, and $n$ is the image of itself.

    \textbf{(Coverage constraint)} Every $i$ in an interval $[x_{l-1}+1, x_l]$ satisfies $i \leq x_l = \varphi_X^n(i)$ by definition of the interval. The same holds for the first interval $[1, x_1]$ and the last interval $[x_{k-1}+1, n]$. Hence $i \leq \varphi_X^n(i)$ for all $i \in [1,n]$.
\end{proof}

\begin{lemma}\label{lem:normalisation-equivalence}
    Given an optimal specialization function $\varphi \in \arg \min_{\varphi \in \Phi^k_n} \textit{cost}(\varphi)$ with range $X\cup\{n\}$ where $X$ does not contain $\{n\}$,
    \begin{equation}
        \textit{cost}(\varphi) = \textit{cost}(\varphi_X^n).
    \end{equation}
\end{lemma}
\begin{proof}
    The case $k = 1$ is immediate: both $\varphi$ and $\varphi_X^n$ are the constant map $i \mapsto n$, so their costs are equal.
    For $k > 1$, let $X = \{x_1 < \cdots < x_{k-1}\} \subseteq [1,n-1]$ be the range of $\varphi$ excluding $n$. We show that $\textit{cost}(\varphi_X^n) \leq \textit{cost}(\varphi)$; since $\varphi$ is optimal, this implies equality.
    We compare the terms $s(\varphi(i))f(i)$ and $s(\varphi_X^n(i))f(i)$ for each $i \in [1,n]$, considering three cases.
    \begin{enumerate}
        \item \textbf{$\varphi(i) = \varphi_X^n(i)$:} The terms coincide trivially.
        \item \textbf{$\varphi(i) < \varphi_X^n(i)$:} We show this case is impossible.
              For any interval defining $\varphi_X^n$, the smallest element of $X \cup \{n\}$ that is greater than or equal to $i$ is assigned as $\varphi_X^n(i)$.
              Any element of $X \cup \{n\}$ strictly smaller than $\varphi_X^n(i)$ is at most $x_{l-1}$ for some $l$, but then $\varphi(i) \leq x_{l-1} < i$, violating the coverage constraint $i \leq \varphi(i)$.
        \item \textbf{$\varphi(i) > \varphi_X^n(i)$:} Since $\varphi_X^n(i) \leq \varphi(i)$, monotonicity of $s$ gives $s(\varphi_X^n(i)) \leq s(\varphi(i))$, so the contribution of $i$ to $\textit{cost}(\varphi_X^n)$ is no greater than its contribution to $\textit{cost}(\varphi)$.
    \end{enumerate}
    Summing over all $i \in [1,n]$, we obtain $\textit{cost}(\varphi_X^n) \leq \textit{cost}(\varphi)$. Since $\varphi$ is optimal, equality holds.
\end{proof}

\begin{theorem}\label{thm:representation}
    \begin{equation}
        \min_{\varphi \in \Phi^k_n} \textit{cost}(\varphi) = \min_{X \in \mathbb{X}^{k-1}_{n-1}} \textit{cost}(\varphi_X^n)
    \end{equation}
\end{theorem}
\begin{proof}
    The inequality $\min_{X \in \mathbb{X}^{k-1}_{n-1}} \textit{cost}(\varphi_X^n) \geq \min_{\varphi \in \Phi^k_n} \textit{cost}(\varphi)$ follows from Lem\-ma~\ref{lem:norm-subset}, which shows that every normalized specialization function $\varphi_X^n$ lies in $\Phi^k_n$, so the minimum over the smaller family cannot be smaller.
    For the reverse inequality, let $\varphi \in \arg\min_{\varphi \in \Phi^k_n} \textit{cost}(\varphi)$ be an optimal specialization function with range $X \cup \{n\}$. By Lemma~\ref{lem:normalisation-equivalence}, the normalized function $\varphi_X^n$ satisfies $\textit{cost}(\varphi_X^n) = \textit{cost}(\varphi)$. Hence $\min_{X \in \mathbb{X}^{k-1}_{n-1}} \textit{cost}(\varphi_X^n) \leq \textit{cost}(\varphi) = \min_{\varphi \in \Phi^k_n} \textit{cost}(\varphi)$.
    Combining both inequalities proves the theorem.
\end{proof}
Searching for an optimal normalized specialization in the set $\mathbb{X}^{k-1}_{n-1}$ also gives an optimal specialization function.
Reducing the search space provides a handle for formulating the optimization problem as a dynamic program.

\paragraph{Interval Costs.}
With a normalized specialization function $\varphi_X^n$, the domain $[1,n]$ of node types is partitioned into $k$ disjoint intervals $[1,x_1]$, $[x_1+1, x_2]$, $\ldots$, $[x_{k-1}+1, n]$.
An interval $[i,j]$ with $1\leq i \leq j \leq n$ would represent node types with $i, i+1, \ldots, j$ values/children that choose the specialization of node type $j$.
The elements of an interval $[i,j]$ are mapped to the same specialization, which is the last element $j$ in the interval. An element in interval $[1,x_1]$ is mapped to specialization $x_1$, an element in $[x_{l-1}+1, x_l]$ is mapped to $x_l$, for all $l$, $1 < l < k$ and an element in $[x_{k-1}+1, n]$ is mapped to $n$.
Hence, we can compute the costs for node types in a domain interval $[i,j]$ as follows,
\begin{equation}
    c(i,j) = \sum_{l=i}^{j} s(j) f(l) =  s(j) \sum_{l=i}^{j} f(l)
\end{equation}
since the node types $l \in [i,j]$ use the specialization $j$.
Given a normalized specialization function $\varphi_X^n$ with its domain partitioning $[1, x_1], [x_1+1, x_2], \ldots, [x_{k-2}+1, x_{k-1}], [x_{k-1}+1, n]$ we can express the costs over its disjoint sets
\begin{equation*}
    \textit{cost}(\varphi_X^n) = c(1,x_1) + c(x_1+1, x_2) + \ldots + c(x_{k-2}+1, x_{k-1}) + c(x_{k-1}+1, n)
\end{equation*}

\begin{lemma}[Cost Equivalence]\label{lem:cost-equivalence}
    Let $\varphi_X^n$ be a normalized specialization function with $X = \{x_1, \ldots, x_{k-1}\} \subseteq [1, n-1]$.
    Then the interval-sum formula and the element-wise definition of cost coincide:
    \[
        c(1, x_1) + \sum_{l=2}^{k-1} c(x_{l-1}+1,\, x_l) + c(x_{k-1}+1,\, n)
        \;=\;
        \sum_{i=1}^{n} s\!\left(\varphi_X^n(i)\right) f(i).
    \]
\end{lemma}
\begin{proof}
    The intervals $[1, x_1],\; [x_1+1, x_2],\; \ldots,\; [x_{k-1}+1, n]$ form a partition of $[1, n]$ into $k$ consecutive disjoint sets whose union is $[1, n]$.
    On each interval $[x_{l-1}+1, x_l]$ (with the convention $x_0 = 0$ and $x_k = n$), the normalized specialization function assigns the constant value $\varphi_X^n(i) = x_l$
    to every $i$ in that interval.
    Therefore,
    \begin{align*}
        \sum_{i=1}^{n} s\!\left(\varphi_X^n(i)\right) f(i)
         & = \sum_{l=1}^{k} \sum_{i=x_{l-1}+1}^{x_l} s\!\left(\varphi_X^n(i)\right) f(i) \\
         & = \sum_{l=1}^{k} \sum_{i=x_{l-1}+1}^{x_l} s(x_l)\, f(i)                       \\
         & = \sum_{l=1}^{k} s(x_l) \sum_{i=x_{l-1}+1}^{x_l} f(i)                         \\
         & = \sum_{l=1}^{k} c(x_{l-1}+1,\, x_l),
    \end{align*}
    which is precisely the interval-sum formula.
\end{proof}

To avoid computing the frequency counts $\sum_{i\leq l \leq j} f(l)$ for each interval $[i,j]$, the cumulative frequency counts can be deduced by a recursive schema shown below:
\begin{align*}
    F_0 & = 0                             \\
    F_1 & = f(1)                          \\
    F_i & = F_{i-1} + f(i) & i \in [2, n] \\
\end{align*}
where $F_i=\sum_{j=1}^i f(j)$ for all $i$, $1 \leq i \leq n$.
The cumulative frequency count gives us the following identity:
\begin{equation}
    \sum_{i \leq l \leq j} f(l) = F_j - F_{i-1}
\end{equation}
for all $i,j$ where $1 \leq i \leq j \leq n$.
The costs can be reduced to the following calculation
\begin{equation}
    c(i,j)  = s(j)\,(F_j - F_{i-1})
\end{equation}
The computation for each interval becomes constant.
We have $\frac{n(n+1)}{2}$ intervals in the range from $1$ to $n$ and $n$ steps for computing the cumulative frequency count.
Hence, the worst-case runtime for computing the costs of all intervals is $\mathcal{O}\!\left(n+\frac{n(n+1)}{2}\right)=\mathcal{O}(n^2)$.

\paragraph{Dynamic Program.}
We construct a two-dimensional matrix $Y$ that stores the minimum cost of partitioning the node types from $1$ to $n$ into at most $k$ consecutive disjoint intervals that cover the entire node range.
The $j$-th row represents the partitioning from node type $1$ to $j$, and the $l$-th column represents an $l$-partitioning with $l$ consecutive disjoint sets in the range from $1$ to $j$ covering the whole range.
A disjoint interval represents a specialization for node types in the interval, using the last element in the interval as its specialization (cf.\ Definition~\ref{def:normalized-specialization}).

The recurrence is:
\begin{align*}
    y_{j,1}    & = c(1,j)                                                 & j \in [1,n]                  \\
    y_{1,l}    & = M                                                      & l \in [2,k]                  \\
    y_{j, l+1} & = \min_{1 \leq i < j} \left( y_{i, l} + c(i+1,j) \right) & j \in [2,n],\, l \in [1,k-1]
\end{align*}
where $M$ is a sufficiently large constant representing an infeasible solution~\cite{mathprog02}.
A single element cannot be split into more than one non-empty interval, hence $y_{1,l} = M$ for $l > 1$.
The optimal cost of a $k$-partition of $[1, n]$ is given by $y_{n,k}$.

To establish correctness, we first verify that the problem has the optimal substructure property~\cite{cormen09}.
\begin{lemma}[Optimal Substructure]
    Let $\varphi_X^n$ with $X = \{x_1, \ldots, x_{k-1}\}$ be an optimal solution for the range $[1,n]$ with $k$ intervals and $k>1$.
    Then the restriction $X' = \{x_1, \ldots, x_{k-2}\}$ is an optimal normalized specialization $\varphi_{X'}^{x_{k-1}}$ for the range $[1, x_{k-1}]$ with $k-1$ intervals.
\end{lemma}
\begin{proof}
    Suppose for contradiction that $\varphi_{X'}^{x_{k-1}}$ is not optimal for the range $[1, x_{k-1}]$ with $k-1$ intervals, and let $\varphi_{\tilde{X}'}^{x_{k-1}}$  be a strictly cheaper solution.
    Then $\tilde{X} = \tilde{X}' \cup \{x_{k-1}\}$ yields a $k$-interval solution for $[1,n]$ with cost
    \[
        \textit{cost}(\tilde{X}') + c(x_{k-1}+1, n) < \textit{cost}(X') + c(x_{k-1}+1, n) = \textit{cost}(X),
    \]
    contradicting the optimality of $X$.
\end{proof}

\begin{theorem}[Correctness of $\mathbf{Y}$]
    For all $j \in [1,n]$ and $l \in [1,k]$, the entry $y_{j,l}$ equals the minimum cost of any $l$-partition of $[1,j]$.
\end{theorem}
\begin{proof}
    By induction on $j$.

    \textbf{Base case ($j = 1$).} The unique $1$-partition of $[1,1]$ has cost $c(1,1)$, so $y_{1,1} = c(1,1)$ is correct. For $l > 1$, no $l$-partition of a single element exists, so $y_{1,l} = M$ is correct.

    \textbf{Inductive step ($j > 1$).} Assume the claim holds for all $j' < j$. For $l = 1$, the unique $1$-partition of $[1,j]$ is the single interval with cost $c(1,j)$, which is correctly recorded by $y_{j,1}$. For $l > 1$, every $l$-partition of $[1,j]$ consists of an $(l-1)$-partition of $[1,i]$ for some $1 \leq i < j$, followed by the interval $[i+1,j]$. By the induction hypothesis, $y_{i,l-1}$ gives the minimum cost of any $(l-1)$-partition of $[1,i]$. The recurrence
    \[
        y_{j,l} = \min_{1 \leq i < j}\!\left(y_{i,l-1} + c(i+1,j)\right)
    \]
    therefore yields the minimum cost among all such decompositions, completing the induction.
\end{proof}

To reconstruct the actual optimal $k$-partition, we compute witness matrix $\mathbf{W}$ alongside $\mathbf{Y}$ for the recursive case that contains following witnesses:
\begin{align*}
    w_{j, l+1} & = \arg \min_{1 \leq i < j} \left ( y_{i, l} + c(i+1,j) \right) & j \in [2,n], l \in [1,k-1]
\end{align*}
where $\arg \min$ returns the smallest $i$ among all minimal indices, and other entries of the matrix can be set to arbitrary values, such as zero.

The specializations in the set $X$ for the normalized specialization function $\varphi_X^n$ with $X=\{x_1, \ldots, x_{k-1}\}$ are then recovered from the witnesses by back-tracking:
\begin{align}
    x_{k-1} & = w_{n,k}                           \\
    x_{l}   & = w_{x_{l+1}, l+1} & l \in [1, k-2]
\end{align}

\begin{lemma}[Witness–Partition Correspondence]\label{lem:witness-partition}
    Let $\varphi_X^n$ with $X = \{x_1, \ldots, x_{k-1}\}$ be an optimal normalized
    specialization function recovered by back-tracking through $\mathbf{W}$:
    \begin{align*}
        x_{k-1} & = w_{n,\,k},                               \\
        x_l     & = w_{x_{l+1},\,l+1}, \quad l \in [1, k-2].
    \end{align*}
    Then for all $l \in [1, k-1]$, the witness $w_{x_{l+1},\, l+1}$ equals $x_l$, i.e.,
    \[
        w_{x_{l+1},\, l+1} = x_l.
    \]
\end{lemma}
\begin{proof}
    We proceed by downward induction on $l$, from $l = k-1$ down to $l = 1$.

    \textbf{Base case ($l = k-1$).}
    By the definition of the recurrence,
    \[
        y_{n,\,k} = \min_{1 \leq i < n}\!\left( y_{i,\,k-1} + c(i+1,\, n) \right),
    \]
    and the witness records
    \[
        w_{n,\,k} = \arg\min_{1 \leq i < n}\!\left( y_{i,\,k-1} + c(i+1,\, n) \right).
    \]
    By the Correctness Theorem, $y_{n,k}$ equals the minimum cost of any $k$-partition
    of $[1,n]$.
    The optimal normalized specialization $\varphi_X^n$ achieves this minimum, and its last interval is $[x_{k-1}+1, n]$, so the split point $i = x_{k-1}$ is a minimizer of the recurrence.
    Hence $w_{n,k} = x_{k-1}$.

    \textbf{Inductive step.}
    Suppose the claim holds for $l+1$, i.e., $w_{x_{l+2},\, l+2} = x_{l+1}$, so that
    back-tracking has correctly recovered $x_{l+1}$.
    We show $w_{x_{l+1},\, l+1} = x_l$.

    By the recurrence,
    \[
        y_{x_{l+1},\, l+1}
        = \min_{1 \leq i < x_{l+1}}\!\left( y_{i,\,l} + c(i+1,\, x_{l+1}) \right),
    \]
    with witness
    \[
        w_{x_{l+1},\, l+1}
        = \arg\min_{1 \leq i < x_{l+1}}\!\left( y_{i,\,l} + c(i+1,\, x_{l+1}) \right).
    \]
    By the Optimal Substructure Lemma applied repeatedly from the full range $[1,n]$
    down to the prefix $[1, x_{l+1}]$, the restriction $X' = \{x_1, \ldots, x_{l-1}\}$
    is an optimal $l$-partition of $[1, x_{l+1}]$ (using $x_l$ as the last split point
    before $x_{l+1}$).  Therefore the split $i = x_l$ achieves
    \[
        y_{x_l,\, l} + c(x_l + 1,\, x_{l+1}) = y_{x_{l+1},\, l+1},
    \]
    so $i = x_l$ is a minimizer of the recurrence, and hence
    \[
        w_{x_{l+1},\, l+1} = x_l. \qedhere
    \]
\end{proof}

\paragraph{Worst-Case Runtime Complexity.}
Computing $\mathbf{Y}$ requires filling $n \times k$ entries, each requiring at most $n$ comparisons, giving an overall complexity of $\mathcal{O}(kn^2)$.
If $k = \mathcal{O}(n)$, this reduces to $\mathcal{O}(n^3)$.
Precomputing all interval costs takes $\mathcal{O}(n^2)$, which is dominated by the recurrence and hence can be neglected in the worst-case runtime complexity.

\paragraph{Extension.}
The coverage constraint in Definition~\ref{def:specialization} permits the construction of a monotonic space function if some specializations with higher node-types are more space efficient than lower specializations.
For example, we use a dense node representation (storing up to $n$ children/values) in addition to a sparse representation that can have up to $i$ values/children, for all $i$, $1 \leq i \leq n$.
However, the dense implementation becomes more space-efficient at a certain number of values/children, since the dense representation encodes its values/children without storing their indices, whereas the sparse specialization requires indices.
In our implementation, after the specializations for node type 241, the dense representation for 256 node types becomes more efficient.
To ensure that the dynamic program can still be used, we construct a new monotonic space function $s': [1,n] \rightarrow \mathbb{N}$ from the original $s$ by
\[
    s'(i) \mapsto \min_{i \leq l \leq n} s(l),
\]
which selects the cheapest available specialization for each node type $i$.
If no cheaper representation exists, $s'(i) = s(i)$. This local choice does not affect the optimality of the specialization.

The construction of the new monotonic space function $s'\colon [1,n] \to \mathbb{N}$ from the original function $s$ by $s'(i) = \min_{i \leq l \leq n} s(l)$ selects the cheapest available specialization for each node type $i$.
If no cheaper representation exists, $s'(i) = s(i)$.
This local choice does not affect the optimality of the specialization.
The function $s'$ is monotone: for any $i \leq j$ in $[1,n]$, every index eligible when computing $s'(j)$ is also eligible when computing $s'(i)$, since the minimization for $s'(i)$ ranges over a super set of that for $s'(j)$.
Taking a minimum over a larger set can only decrease or preserve the value, so $s'(i) \leq s'(j)$.

For the optimization problem, the solution $\varphi$ obtained using $s'$ must be re-interpreted in a post-processing step.
The actual specialization $\varphi'$ remaps each specialization chosen under $s'$ to the corresponding minimizer under the original function $s$:
\[
    \varphi'(i) \mapsto \arg \min_{\varphi(i) \leq l \leq n} s(l)
\]

Another potential optimization is to consider the zero node where neither nodes nor values are stored.
Zero nodes may arise from changes to inner nodes and leaves and require special care, though they can be handled programmatically and do not affect the optimization problem.

\section{Implementation\label{sec:implementation}}

Our database uses a layered architecture in which each layer builds on the one below and exposes well-defined interfaces, mirroring the natural layering of the Verkle Trie specification itself.
The architecture comprises five layers: (1) Interface to Transaction Manager, (2) Tree Implementation, (3) Node Manager, (4) Storage, and (5) File I/O. Our architecture supports concurrent read and write operations.

\begin{figure*}[htbp]
    \centering
    \includegraphics[width=\textwidth]{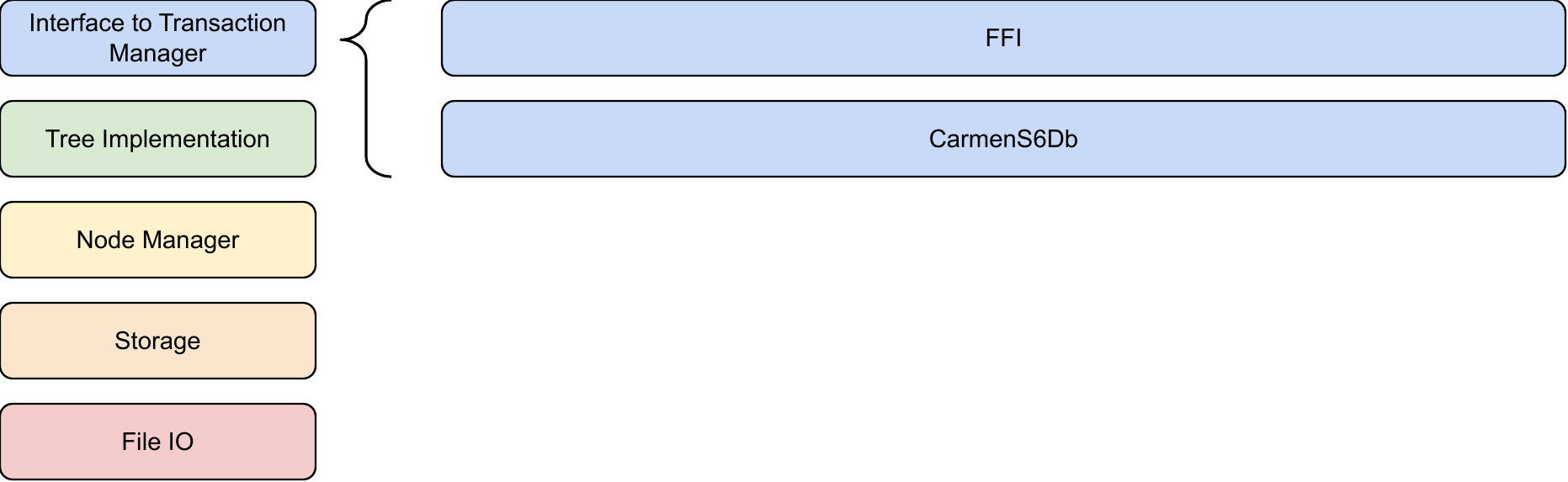}
    \caption{Layered Architecture - Interface to Transaction Manager}
    \label{fig:carmen-layers-2}
\end{figure*}

\paragraph{Transaction Manager Interface.}
The highest layer is the transaction manager interface, which provides the EVM with a way to query and update the database.
The interface is split into the database and the state interfaces.
The former provides functions for opening and closing the DB, creating checkpoints, and accessing live and archive states.
This is implemented in \href{https://github.com/0xsoniclabs/carmen/blob/main/rust/src/lib.rs}{lib.rs}.
The state interface provides functions for querying and updating the state, such as getting and setting account balances, nonces, storage slots, and code, and is implemented in \href{https://github.com/0xsoniclabs/carmen/blob/main/rust/src/database/verkle/state.rs}{database/verkle/state.rs}.
Because the transaction manager, as well as other components like the client, consensus, the VM, RPC, etc., are implemented in Go, the database and state interfaces are exposed to Go via FFI implemented in the file \href{https://github.com/0xsoniclabs/carmen/blob/main/rust/src/ffi/exported.rs}{ffi/exported.rs}.

\begin{figure*}[htbp]
    \centering
    \includegraphics[width=\textwidth]{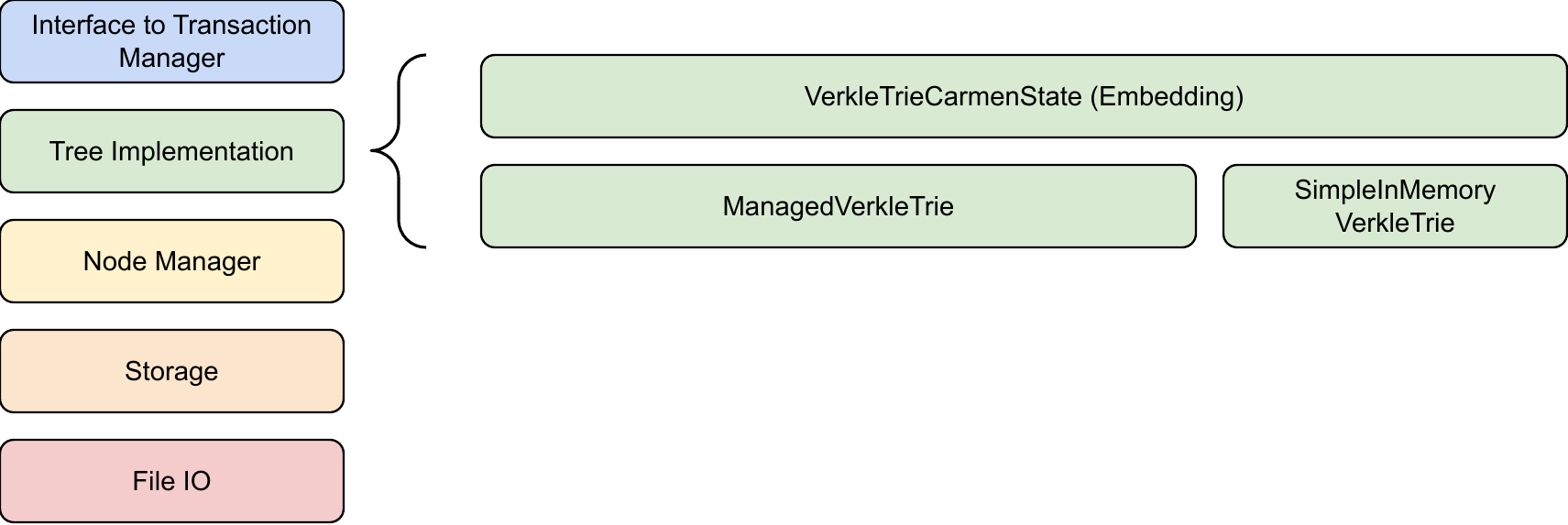}
    \caption{Layered Architecture - Tree Implementation}
    \label{fig:carmen-layers-1}
\end{figure*}

\paragraph{Tree Implementation}.
The tree implementation provides an interface for reading and writing key-value pairs stored in the Verkle Trie and for computing root commitments.
It consists of two sub-layers: the embedding layer and the Verkle Trie layer.
The embedding layer maps the Ethereum state to key-value pairs.
These key-value pairs are then stored using the Verkle Trie layer.
Computing the key for a value is expensive, as it requires computing a Pedersen commitment. Therefore, we memoize computed keys to avoid redundant computations.
The embedding can be found in \href{https://github.com/0xsoniclabs/carmen/tree/main/rust/src/database/verkle/embedding}{database/verkle/embedding}.
There are two implementations of the verkle trie layer: an in-memory implementation for testing and a main \texttt{managed-trie} implementation that uses the node manager to access nodes and the storage layers to persist them.
The in-memory implementation can be found in \href{https://github.com/0xsoniclabs/carmen/tree/main/rust/src/database/verkle/variants/simple}{database/verkle/variants/simple} and the \texttt{managed-trie} implementation in \href{https://github.com/0xsoniclabs/carmen/tree/main/rust/src/database/verkle/variants/managed}{database/verkle/variants/managed}.
Generic tooling for working with managed tries, independent of the managed Verkle Trie, can be found in \href{https://github.com/0xsoniclabs/carmen/tree/main/rust/src/database/managed_trie}{database/managed\_trie}.

\begin{figure*}[htbp]
    \centering
    \includegraphics[width=\textwidth]{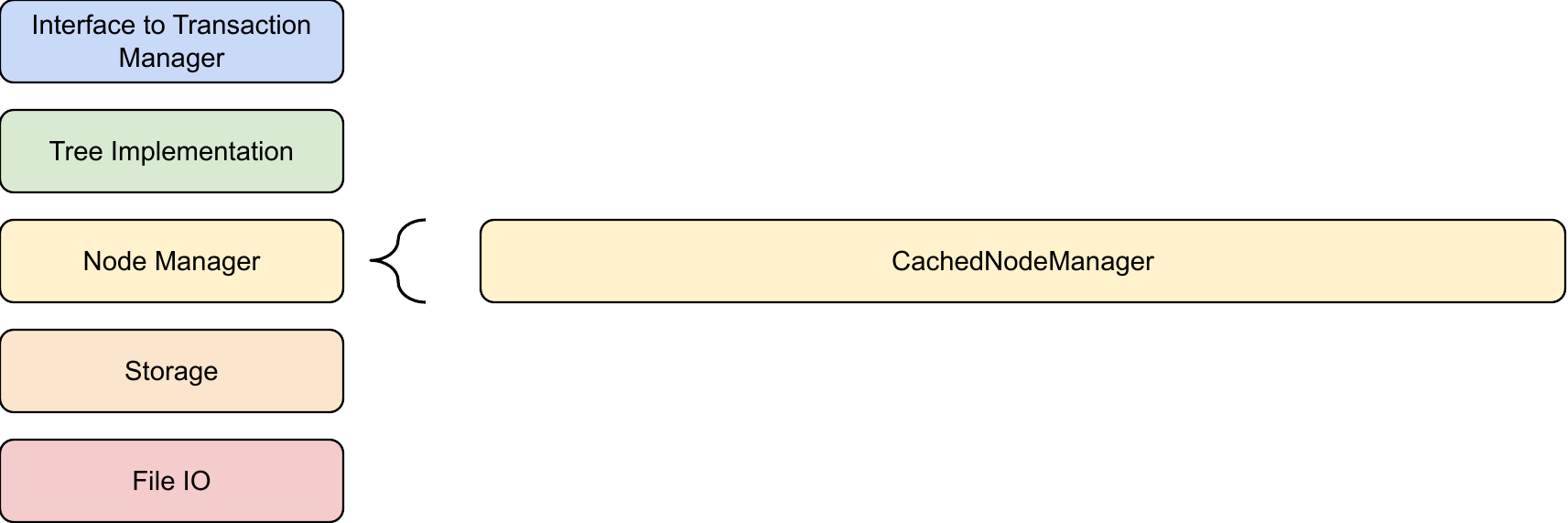}
    \caption{Layered Architecture - Node Manager}
    \label{fig:carmen-layers-3}
\end{figure*}

\paragraph{Node Manager.}
The node manager owns all in-memory node instances and provides functions to create, delete, and access nodes.
The node manager ensures that only a single copy of each node exists in memory at any time and provides access to them via read and write guards.
It is an interface, of which there are two implementations:
an \textit{in-memory} node manager, which keeps all nodes in memory, used only for testing,
and a \textit{cached} node manager with limited capacity that evicts nodes, writing them to disk via the storage layer.

The \textit{cached} node manager uses \texttt{quick\_cache} \cite{quickcache}, a high-performance concurrent cache for Rust, as its cache.
It uses an eviction algorithm derived from CLOCK-PRO \cite{clock-pro} and optimized for weak access patterns, such as loops and scans.
The main difference with the well-known LRU algorithm is that, instead of using the \textit{recency} of an item, i.e., the distance from the last inserted element for deciding which element to evict, it uses a metric called \textit{reuse distance}, defined as the distance from the last time that specific element was requested.
This allows \texttt{quick\_cache} to avoid inserting elements that are requested only once in the cache, as their reuse distance is infinite.
The main implication is that inserting an element into a full cache does not guarantee that it will be inserted into the main cache allocation.
Instead, \texttt{quick\_cache} inserts it into a secondary allocation called the \textit{ghost pool}, which tracks previously requested elements.
Whenever an element is requested and not present in the cache, it is looked up in the ghost pool to calculate its reuse distance and, if it is not found, inserted into the main cache allocation.
However, this behavior is problematic for the node manager, as node ownership is delegated to the manager; therefore, it must ensure that the node remains in the cache after insertion.
Therefore, elements are explicitly pinned before insertion.
\newline
Much of the complexity is hidden in the \texttt{LockCache} component that wraps \texttt{quick\_cache} and provides  read\/write access guards to its elements,
which are stored in a fixed-size array as \texttt{RWLock}, while \texttt{quick\_cache} stores the array occupied slots.

This decoupling is necessary because stable storage locations are required to hand out lock guards whose lifetimes are bound by the cache's lifetime.
\texttt{quick\_cache} only returns copies of its internal elements, and therefore, a guard to it would refer to a temporary object.
This structure has the great advantage that locking happens only in the \texttt{LockCache}, which concentrates all synchronization logic in a single place.
Moreover, because only guards are given away, owning a write guard guarantees that no other reference to the node exists.
This is especially important when updating the trie, as it allows for safe modifications of the tree.
\newline

The \texttt{CachedNodeManager} implements the Node Manager interface, wrapping the \texttt{LockCache} and a storage backend.
When a node is requested, the node manager first checks the cache; if it is not found, it retrieves it from the storage backend and caches it.
It is also responsible for keeping the storage layers up to date when nodes are modified or deleted by flushing the changes to disk only when changes on the node are detected, i.e., when a write guard is dropped.  All of these components can be found in \href{https://github.com/0xsoniclabs/carmen/tree/main/rust/src/node_manager}{node\_manager}.

\begin{figure*}[htbp]
    \centering
    \includegraphics[width=\textwidth]{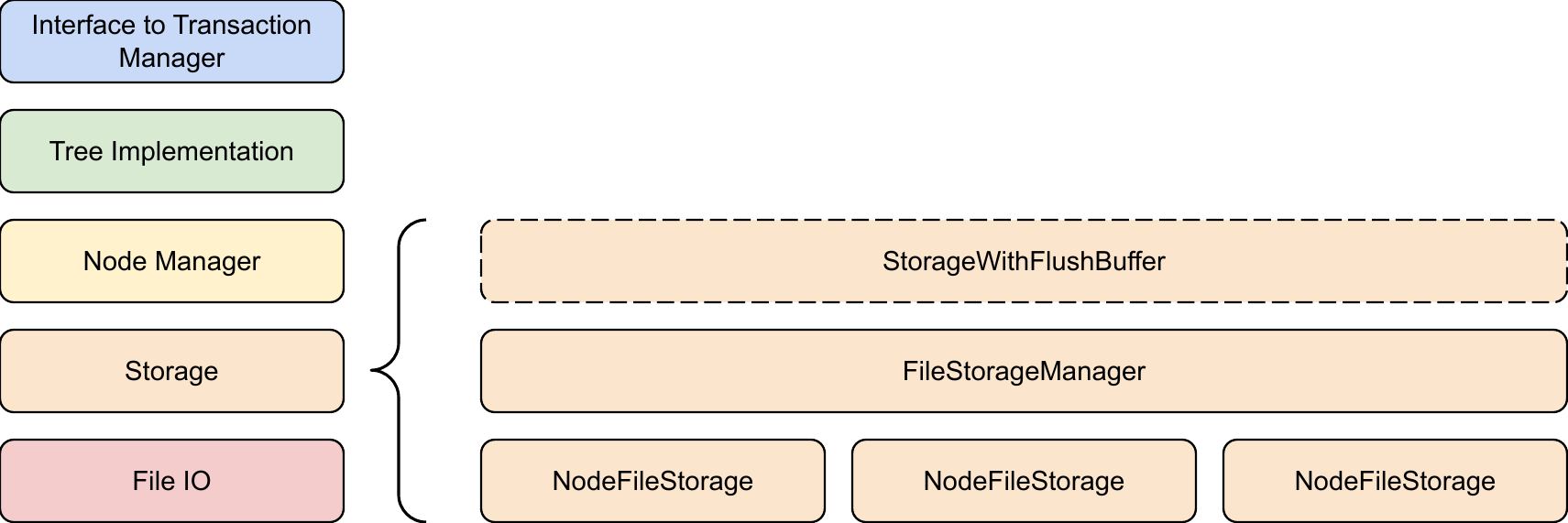}
    \caption{Layered Architecture - Storage}
    \label{fig:carmen-layers-4}
\end{figure*}

\paragraph{Storage Layers.}
The storage layers sit on top of the file layers and are responsible for persisting nodes to disk and loading them into memory upon request.
The nodes are stored in flat files as plain old data.
Nodes have to implement an interface that allows them to be converted to and from bytes.
This way, it is possible to use the in-memory representation of the nodes directly, which incurs minimal overhead, or (as is the case for the main implementation) convert them to a more compact format and serialize this format to disk.
There is one file per node type, so all nodes stored in a given file have the same size.
This enables fast random access by computing a node's offset from its ID.
Each node ID encodes both the file index and the node type.
The offset can be obtained by extracting the index from the ID and then multiplying it by the node type size.

NodeFileStorage manages a file for a specific node type and provides functions to read and write nodes to that file via the file layer.
It also keeps track of the next available offset for new nodes and maintains a free list of deleted nodes to reuse their space.
It is implemented in \href{https://github.com/0xsoniclabs/carmen/tree/main/rust/src/storage/file/node_file_storage}{storage/file/node\_file\_storage}

FileStorageManager multiplexes multiple NodeFileStorage instances, one for each node type, and provides a unified interface for reading and writing the node enum.
It also manages checkpoint creation, which ensures the database's consistency and durability.
The implementation of FileStorageManager can be found in \href{https://github.com/0xsoniclabs/carmen/tree/main/rust/src/storage/file/file_storage_manager}{storage/file/file\_storage\_manager}

StorageWithFlushBuffer is a purely optional layer that wraps a storage backend and adds a flush buffer on top of it.
When nodes are modified or deleted, they are not immediately written to disk; instead, they are stored in the flush buffer and asynchronously written to disk by background workers, thereby accelerating cache eviction and ultimately reducing latency.
It also ensures that nodes evicted from the cache can be read from this buffer until they are guaranteed to be persisted to disk.
The implementation can be found in \href{https://github.com/0xsoniclabs/carmen/blob/main/rust/src/storage/storage_with_flush_buffer.rs}{storage/storage\_with\_flush\_buffer.rs}

\begin{figure*}[htbp]
    \centering
    \includegraphics[width=\textwidth]{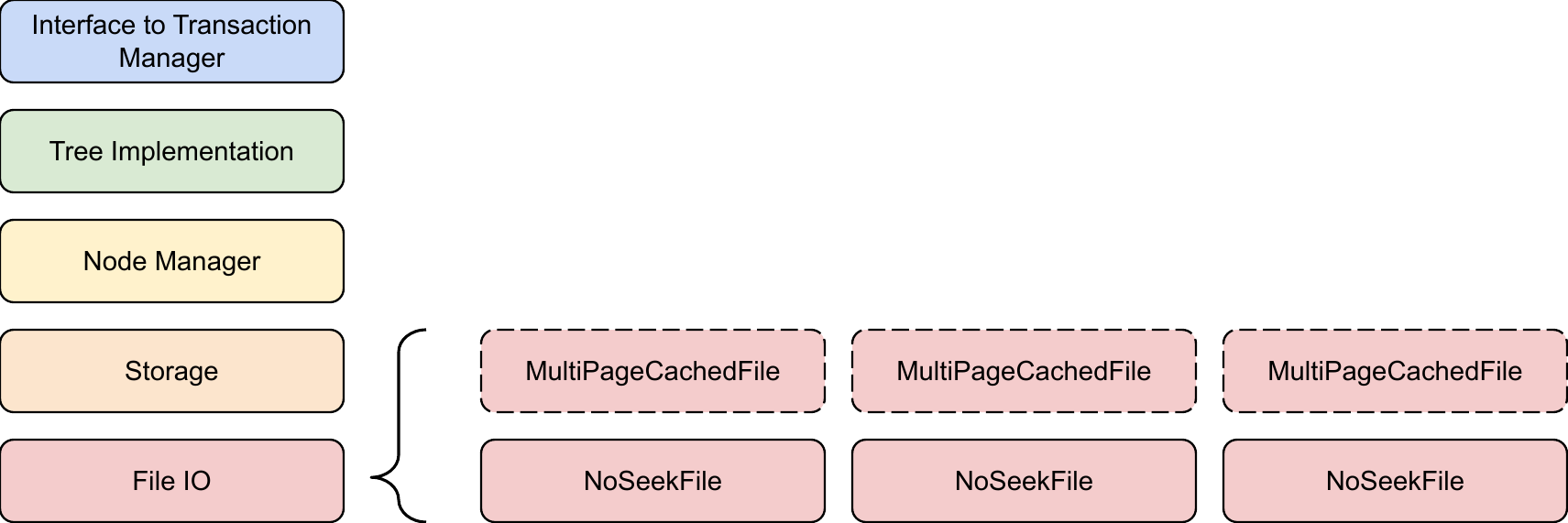}
    \caption{Layered Architecture - File I/O}
    \label{fig:carmen-layers-5}
\end{figure*}

\paragraph{File Layers.}
The file layer is the lowest level and is responsible for reading and writing raw bytes to and from disk.
Two primary implementations are provided, which differ in their support for parallel operations.
\textsc{SeekFile} utilizes the \texttt{read}~\cite{linux-read}, \texttt{write}~\cite{linux-write}, and \texttt{lseek}~\cite{linux-lseek} system calls and is compatible with all platforms.
All accesses are serialized using a mutex, which prevents parallel execution.
\textsc{NoSeekFile} employs the \texttt{pread} and \texttt{pwrite} system calls~\cite{linux-pread}, which are available exclusively on Unix systems.
Because these calls are cursor-independent, concurrent reads and writes to non-overlapping file regions can proceed in parallel.
They also halve the number of system calls by eliminating the need for a preceding seek.
Concurrent access to overlapping offsets is prevented by a fine-grained locking mechanism that maps offsets to mutexes.
Two optional caching proxies, \textsc{PageCachedFile} and \textsc{MultiPageCachedFile}, can wrap either base implementation using a unified interface.
These proxies serve requests from an in-memory buffer containing either a single 4 KiB page or multiple pages, respectively.
When the target data is already cached, no system calls are required.
On a page miss, if the current page is dirty, it is written back before loading the new page.
Both proxies support direct I/O to bypass the operating system page cache. \textsc{PageCachedFile} serializes all accesses using a mutex, whereas \textsc{MultiPageCachedFile} permits concurrent accesses across different pages.
All four implementations, along with the shared interface, are available in the storage file backend module~\cite{carmen-repo}.

\subsection{Trie Update and Commitment Pipeline}

Block finalization requires the database to complete two interdependent operations within Sonic's 300-millisecond block window: applying all state mutations produced by the block's transactions to the Verkle Trie, and recomputing the cryptographic state root commitment over the resulting structure.
The two tasks are not independent: each inner node's commitment depends recursively on its children's commitments, so commitment recomputation cannot begin until the structural update is complete.
Yet together they must fit within the same latency budget that governs block production.
Meeting this budget requires optimizations at every level of the pipeline.
The optimizations range from how mutations are assembled and delivered to the Trie, how the trie traversal is structured to minimize node accesses and synchronization overhead, and how the elliptic-curve arithmetic underlying the Pedersen commitment scheme is executed to exploit the scheme's algebraic structure and the parallelism available in the hardware.

We introduce a three-stage pipeline that processes blocks into an updated state root and applies optimizations at each stage. The first stage is the responsibility of the embedding layer.
The second stage is carried out by the Verkle Trie layer in cooperation with the node manager.
The third stage, commitment recomputation, is performed within the node manager layer.
Each stage feeds directly into the next, and the design of each is constrained by both the requirements it must satisfy and the guarantees it can rely on from the stage preceding it.

\paragraph{Stage 1 — Batch assembly (embedding layer).}

Upon block completion, the transaction manager delivers the full set of state mutations accumulated during the block's execution: account nonces, balances, contract storage slot writes, and code updates.
In a naive implementation, each mutation would be applied to the trie independently, requiring a separate root-to-leaf traversal for each key.
For a block containing $u$ mutations distributed across $p$ distinct leaf nodes, this yields $\mathcal{O}(u \cdot d)$ node accesses, where $d \leq 32$ is the maximum trie depth.
Since $u$ can reach several thousand for a computationally intensive block, the traversal cost alone would represent a significant fraction of the block window.

The embedding layer avoids this by transforming the entire mutation set into a single sorted \emph{update batch} before any trie access occurs.
Each mutation is first passed through the state embedding, which maps it to a $(k, v)$ pair in the trie's 32-byte key space. The resulting pairs are collected into a list and sorted lexicographically by key.
The cost of this preprocessing step is $\mathcal{O}(u \log u)$, which is dominated by the sort.
The embedding itself requires one Pedersen commitment computation per unique account address, with the results memoized to avoid redundant elliptic-curve operations across mutations that touch the same account.

Organizing mutations into a sorted batch before Trie access confers three structural advantages that compound across the subsequent stages.
First, multiple mutations to the same key within a single block are merged during batch assembly, so the trie is presented with at most one write per key per block, and intermediate values produced by intra-block state changes are never written to the trie at all.
Second, because the complete set of mutations destined for any given node is known before that node is first accessed, the node manager can determine in advance whether a structural transformation (e.g., promoting a sparse node to a larger specialization to accommodate additional slots) will be required, and if so, can allocate the target variant directly rather than performing a sequence of incremental promotions.
Third, and most significantly for the ArchiveDB, batching ensures that each node on any root-to-leaf path is subject to at most one copy-on-write operation per block, regardless of how many of its slots are written.
This bounds the number of fresh node identifiers allocated per block to the number of distinct paths traversed, and simplifies the version accounting logic considerably.

\paragraph{Stage 2 — Breadth-first trie traversal (Verkle trie layer).}

The sorted update batch is handed to the root node, which partitions it into sub-batches for each affected child.
Because the batch is sorted lexicographically by key, the first byte of each key determines which child receives each entry, and the partition can be computed by a single linear scan of the batch without copying any data: each sub-batch is represented as a contiguous slice of the original sorted list.
This partitioning recurses level by level, with each node at depth $\ell$ splitting its received sub-batch among its affected children at depth $\ell + 1$, until the sub-batches reach the extension nodes at the leaves.
The total work performed across all partitioning steps is $\mathcal{O}(u \cdot d)$ comparisons, the same asymptotic cost as the naive per-mutation traversal, but with substantially lower constant factors because the sorted order eliminates redundant path re-traversals and because cache locality is improved by processing sibling nodes together.

The traversal proceeds in a \emph{breadth-first, level-wise} order rather than depth-first.
This search order is required by the locking protocol because a node's structural variant may change as a result of its update.
For instance, when an incoming sub-batch causes the node to exceed its current specialization's capacity, necessitating promotion to a larger variant with a new storage identifier, the parent must be informed of the node's new identifier before the parent's own write lock is released.
Breadth-first traversal ensures that the parent's lock is still held when the child's identifier is updated.

Concurrent access during traversal is managed by a \emph{hand-over-hand locking} protocol.
At any point during the traversal, the node manager holds write locks on all nodes at depth $\ell - 1$ that are parents of at least one node currently being processed, and on all nodes at depth $\ell$ that are actively being updated.
Before descending to depth $\ell + 1$, locks are acquired on all nodes at that depth that will be needed.
Once all nodes at depth $\ell$ have been fully processed and their parent identifiers updated, the locks on depth $\ell - 1$ are released.
At any instant, locks therefore span at most three consecutive levels of the Trie.
As soon as the root's write lock is released at the conclusion of the traversal, concurrent read queries to sub-trees that were not modified by the current block can proceed without contention, permitting the RPC layer and the archive reader to serve historical queries in parallel with
the subsequent commitment recomputation.

When a node must be promoted to a larger specialization to accommodate an incoming write, the target specialization size is determined from the sub-batch size before the node is accessed.
The node manager, therefore, allocates the correct variant in a single operation, writes all incoming values directly into it, and releases the previous node identifier to the free list.
This avoids the sequence of intermediate allocations and releases that would occur if mutations were applied one at a time under a per-mutation traversal, and reduces the transient fragmentation of the
free list during high-write blocks.

\paragraph{Stage 3 — Commitment recomputation (node manager layer).}

Once the Trie traversal is complete and all structural updates have been applied, the state root commitment must be recomputed over the set of nodes dirtied during the update.
Commitment recomputation is the computationally dominant stage of the pipeline.
Each Pedersen commitment to a vector of $n$ field elements requires $n$ elliptic-curve scalar multiplications over the Bandersnatch curve, and each such multiplication is roughly two to
three orders of magnitude more expensive than a single SHA-256 evaluation.
For a dense inner node with $n = 256$ children, a full commitment recomputation from scratch would require 256 scalar multiplications, which at current hardware speeds takes on the order of several milliseconds.
This is well above the per-node budget, consistent with the 300-millisecond block window.
Three complementary optimizations are applied to bring the aggregate commitment cost within budget: homomorphic incremental updates, adaptive multi-scalar multiplication windowing, and multi-threaded work-stealing parallelism.

\paragraph{Homomorphic Incremental Updates.}
The Pedersen commitment scheme is additively homomorphic, and this property can be exploited to avoid full recomputation whenever only a subset of a node's slots have changed.
For a commitment $C$ to a vector in which position $i$ changes from $a_i$ to $a_i'$, the updated commitment satisfies
\begin{equation}
    C' = C + (a_i' - a_i) \cdot G_i,
    \label{eq:pedersen-update}
\end{equation}
requiring exactly one scalar multiplication and one group addition per modified slot, independent of the total vector length.
For inner nodes, this update is applied at each position where a child commitment has changed: if $\delta$ children are modified during a block, the inner node's commitment requires $\delta$ scalar multiplications rather than 256, reducing the cost by a factor of $256 / \delta$.
Since empirical data from Sonic's history shows that the number of modified children per inner node per block is typically small (see Section~\ref{sec:evaluation}), this reduction is substantial in practice.

For extension nodes, the two intermediate commitments $C_1$ and $C_2$ are maintained persistently in the node's cached in-memory representation alongside the final node commitment $C_{\mathrm{ext}}$. Whenever a value slot changes, $C_1$ or $C_2$ is updated via Equation~\eqref{eq:pedersen-update}, and $C_{\mathrm{ext}}$ is then updated incrementally from the new intermediate commitment, avoiding the recomputation of the stem scalar $\mathrm{Stem} \cdot G_2$ which, as a product of a 31-byte field element and a fixed generator, is the most expensive single term in Equation~\eqref{eq:leaf-commitment}.
The three cached values — $C_1$, $C_2$, and $C_{\mathrm{ext}}$ — together with the node's slot data constitute the node's in-memory footprint; they are not persisted to disk, and are reconstructed from the on-disk representation on first load.

\paragraph{Adaptive Multi-scalar Multiplication Windowing.}
The scalar multiplications underlying both the incremental update formula and the full recomputation path are computed using a windowed signed multi-scalar multiplication algorithm with Booth recoding~\cite{booth1951signed}.
In this algorithm, each scalar is recoded into a signed digit representation with digit values in $\{-(2^{w-1}-1), \ldots, 2^{w-1}-1\}$ using a window of width $w$ bits, and the required multiples of each generator point are precomputed and stored in a table of size $2^{w-1}$ per generator.
A larger window reduces the number of group operations per scalar at the cost of a larger precomputation table; the optimal window width therefore depends on both the cost of group operations and the number of generators actively involved in a given computation.

The node manager selects the window width adaptively. When only generators $G_1$ through $G_5$ are active, a window of width $w = 16$ is used, reducing the number of group operations to $\lceil 253 / 16 \rceil = 16$ per scalar at the cost of a precomputation table of $5 \times 2^{15} = 163{,}840$ points, which is feasible given that these five generators are fixed and their tables can be computed once at initialisation. This is the case for the four-term extension node commitment $C_{\mathrm{ext}}$ in Equation~\eqref{eq:leaf-commitment} and for the five-term key computation in the embedding layer.
For all other computations, where up to 256 generators may be active, a window of $w = 10$ is used, balancing operation count against table size.
Using $w = 16$ uniformly across all 256 generators would require $256 \times 2^{15} \approx 8.4$ million precomputed points, which exceeds the practical memory budget for a cache-resident node.

\paragraph{Multi-threaded Work-stealing Traversal.}
The commitment updates across all dirty nodes are independent of one another within a given level of the trie: the updated commitment of a node at depth $\ell$ depends on the updated commitments of its children at depth $\ell + 1$, but not on those of its siblings.
This level-wise dependency structure admits a natural parallel decomposition.
The node manager implements this using a \emph{work-stealing thread pool}: starting from the root, it recursively enqueues a commitment update task for each dirty child node, and each task in turn enqueues tasks for its own dirty children.
Worker threads dequeue and execute tasks as they become available; when a thread exhausts its local task queue, it steals tasks from other threads' queues.
A parent task blocks on a lightweight barrier until all its child tasks have completed, then computes its own commitment update from the children's results.
They can be processed concurrently without synchronization beyond the per-node barriers because sibling subtrees are structurally independent.

The work-stealing approach is effective when the dirty set is large enough to meaningfully exploit the hardware parallelism.
When the dirty set is small (e.g., in blocks that touch only a handful of accounts), the overhead of task creation, queuing, and inter-thread synchronization exceeds the computational cost
of the commitment updates themselves, and the parallelism yields a net
slowdown.
The node manager therefore measures the size of the dirty set before committing to the parallel path, and falls back to a sequential depth-first traversal over the dirty nodes when the set size falls below an empirically determined threshold.
This adaptive dispatch ensures that the parallel implementation does not penalize light blocks while still delivering throughput gains for the computationally intensive blocks that dominate Sonic's production workload.

\section{Evaluation\label{sec:evaluation}}

The goal of the experiments is to evaluate SonicDB S6 across two orthogonal dimensions: storage efficiency and throughput.
On the storage side, we quantify how occupancy-aware node specializations and delta nodes reduce the on-disk footprint of both the LiveDB and the ArchiveDB relative to an unoptimized baseline and to the Geth Verkle implementation.
On the throughput side, we measure how fast the database can process Sonic's historical block sequence and compare against a persistent Geth Verkle baseline extended with a LevelDB backend.

\paragraph{Metrics.}
Storage consumption is reported in gibibytes (GiB) at a fixed replay depth of 40 million blocks.
In Ethereum-compatible blockchains, \textit{gas} is the unit of computational work required to execute a transaction or smart contract operation.
Each operation---such as an arithmetic instruction, a storage read, or a contract call---consumes a fixed or dynamic amount of gas defined by the protocol specification~\cite{wood2014ethereum}.
A block is subject to a \textit{gas limit}, which caps the total computational work that can be included in a single block.
Throughput is therefore naturally expressed in terms of gas processed per unit of time rather than transactions per second, since transactions vary widely in their computational complexity and a single transaction may consume anywhere from 21{,}000 gas (a plain transfer) to several million gas (a complex smart contract interaction).
We report throughput in \textit{millions of gas per second} (Mgas/s), computed as the total gas consumed by all transactions in a replay window divided by the elapsed wall-clock time.
This metric captures both the database's ability to sustain high-frequency block ingestion and the workload's computational intensity, making it a more faithful indicator of production performance than block-per-second counts alone.
End-to-end wall-clock replay time is reported as a secondary throughput indicator.

\paragraph{Platform.}
All experiments were conducted on a commodity machine equipped with an AMD Ryzen 5 3600 CPU and 64~GB of RAM, running Go~1.25.7 and rustc~1.92.0.
This hardware configuration is intentionally representative of machines available to ordinary Sonic or Ethereum node operators, ensuring that results are reproducible outside of a datacenter setting.

\paragraph{Claims.}
The experiments are designed to substantiate three claims made in this paper.
First, the dynamic-programming-selected set of 10 node specializations reduces LiveDB storage by 97.8\%, from 1{,}078~GiB to 17.9~GiB, while 512 specializations reduce it further at the cost of measurable throughput degradation.
Second, that delta nodes, combined with node specializations, reduce ArchiveDB storage by 95\%, from 16{,}141~GiB to 809~GiB, with inner delta nodes responsible for the dominant share of that reduction.
Third, that the full SonicDB S6 stack achieves $3.2\times$ higher throughput than the persistent Geth Verkle baseline, while the ArchiveDB remains within approximately 24\% of LiveDB performance---fast enough to keep pace with Sonic's production block rate of 300 milliseconds per block.

\subsection{Node Type Frequency and Choosing k}

\begin{figure}[htbp]
    \begin{subfigure}{0.49\linewidth}
        \centering
        \includegraphics[width=\textwidth]{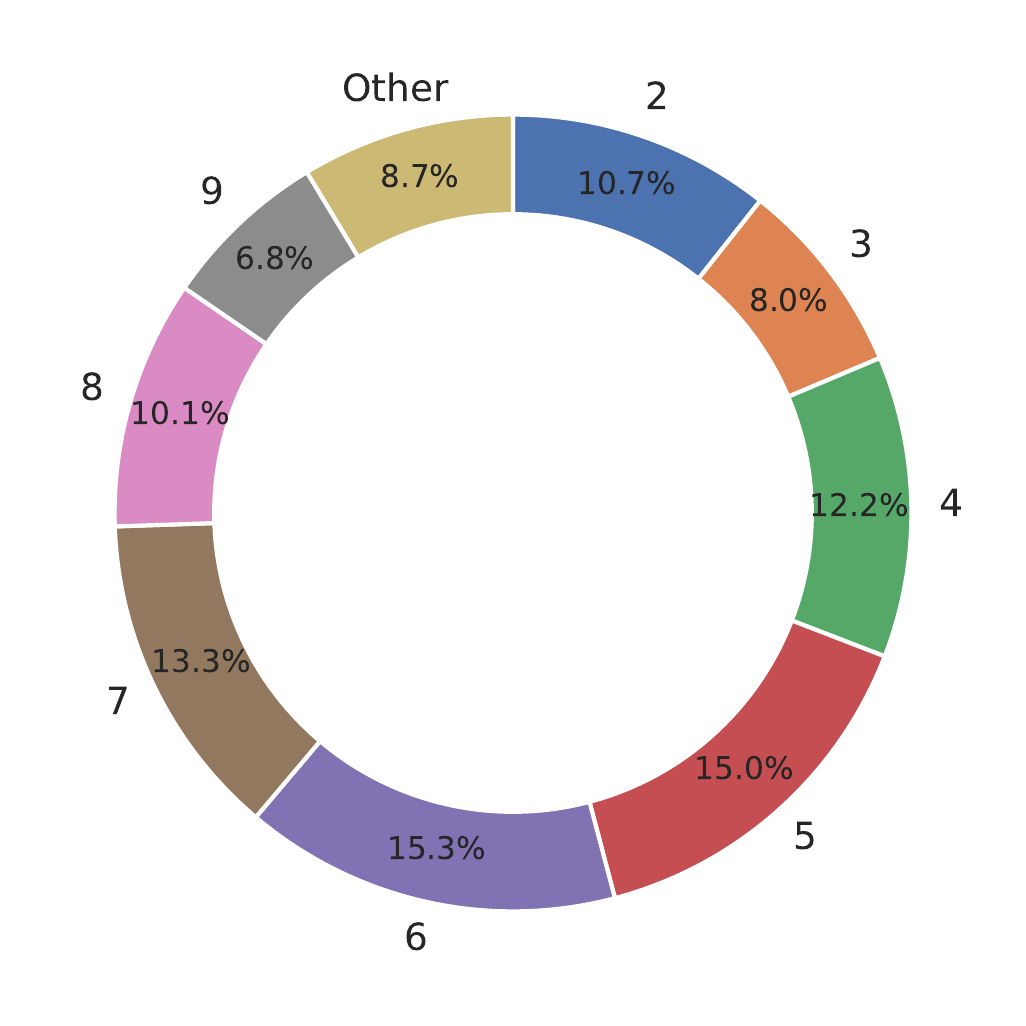}
        \caption{Inner}
        \label{fig:live-inner-node-slot-usage}
    \end{subfigure}
    \hfill
    \begin{subfigure}{0.49\linewidth}
        \centering
        \includegraphics[width=\textwidth]{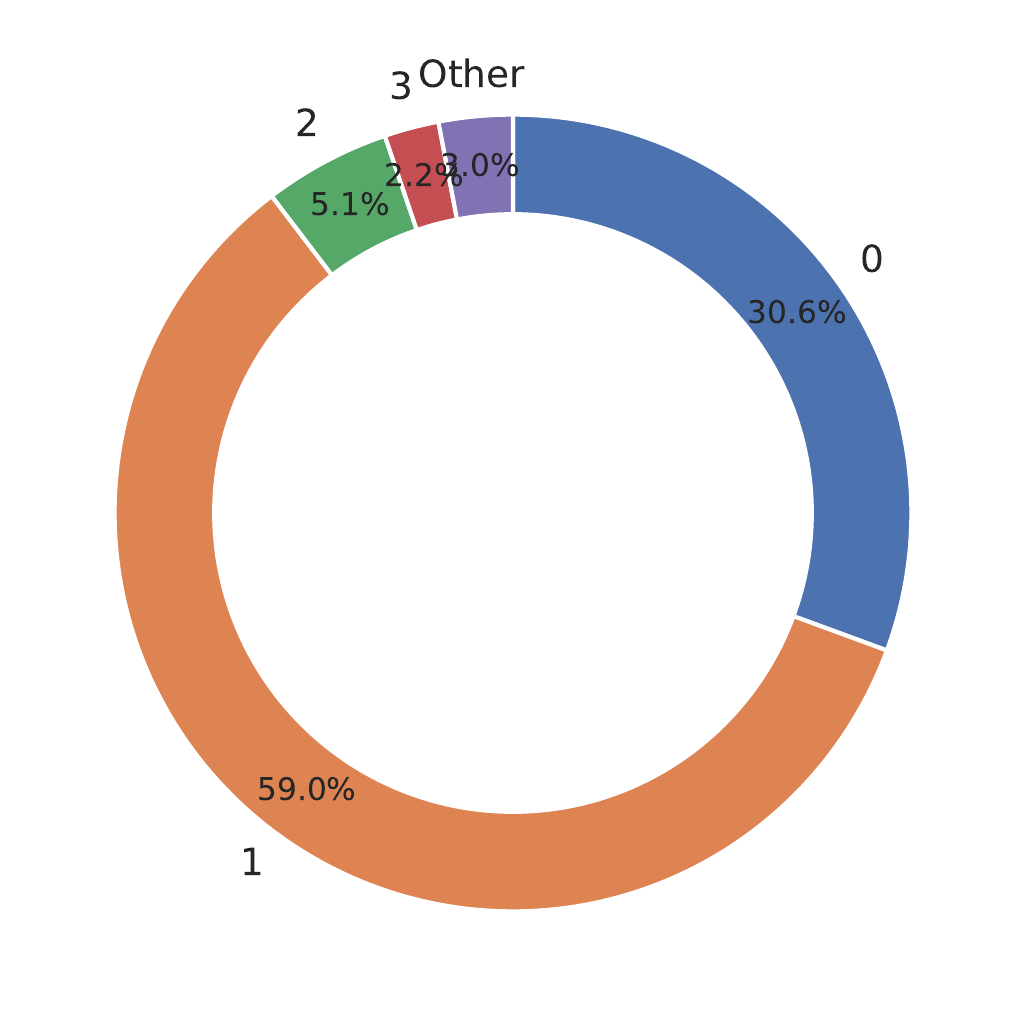}
        \caption{Leaf}
        \label{fig:live-leaf-node-slot-usage}
    \end{subfigure}
    \caption{LiveDB: distribution of actual slot usage in Verkle nodes}
    \label{fig:live-node-slot-usage}
\end{figure}

Figures~\ref{fig:live-node-slot-usage} and~\ref{fig:archive-node-slot-usage} show the distribution of actual slot usage across inner and leaf nodes for both the LiveDB and the ArchiveDB, measured over the first 55 million blocks of Sonic history.
The distributions are heavily skewed toward low occupancy in both cases.
In the LiveDB, 95.7\% of leaf nodes store at most 3 of the 256 values, and 87.6\% of inner nodes have no more than 11 children.
This extreme sparsity confirms that allocating the full 256-slot layout for every node would result in massive storage waste, and motivates the node specialization framework described in Section~\ref{sec:specialization}.
The ArchiveDB exhibits a similarly skewed distribution, though the copy-on-write update semantics introduce a broader tail, as nodes that are frequently modified accumulate more slots over successive versions.
Figure~\ref{fig:total-vs-reusable-node-counts} further illustrates the dynamic interplay between node creation and deletion of inner nodes in the LiveDB over the first 1.4 million blocks.
The figures show that as the trie evolves, nodes are repeatedly promoted to larger specializations when their occupancy grows, and their former storage slots are released to the free list rather than immediately reclaimed.

This behavior produces a characteristic pattern in which the number of reusable offsets rises sharply whenever a wave of specialization upgrades occurs, temporarily widening the gap between total allocated slots and slots that are actively storing data. There is a discrepancy: instead of being freed, nodes are put into free lists so that future node instantiations can reuse them. The figure shows that this gap can be substantial, underscoring that the choice of specialization granularity directly affects not only the per-node storage footprint but also the efficiency of space recycling.
Taken together, these observations establish the empirical foundation for the storage optimizations presented in this paper: the Trie is sparse, node occupancy is non-uniformly distributed, and the number of distinct specializations must be carefully balanced to avoid excessive fragmentation and reuse overhead.

\begin{figure}[htbp]
    \begin{subfigure}{0.49\linewidth}
        \centering
        \includegraphics[width=\textwidth]{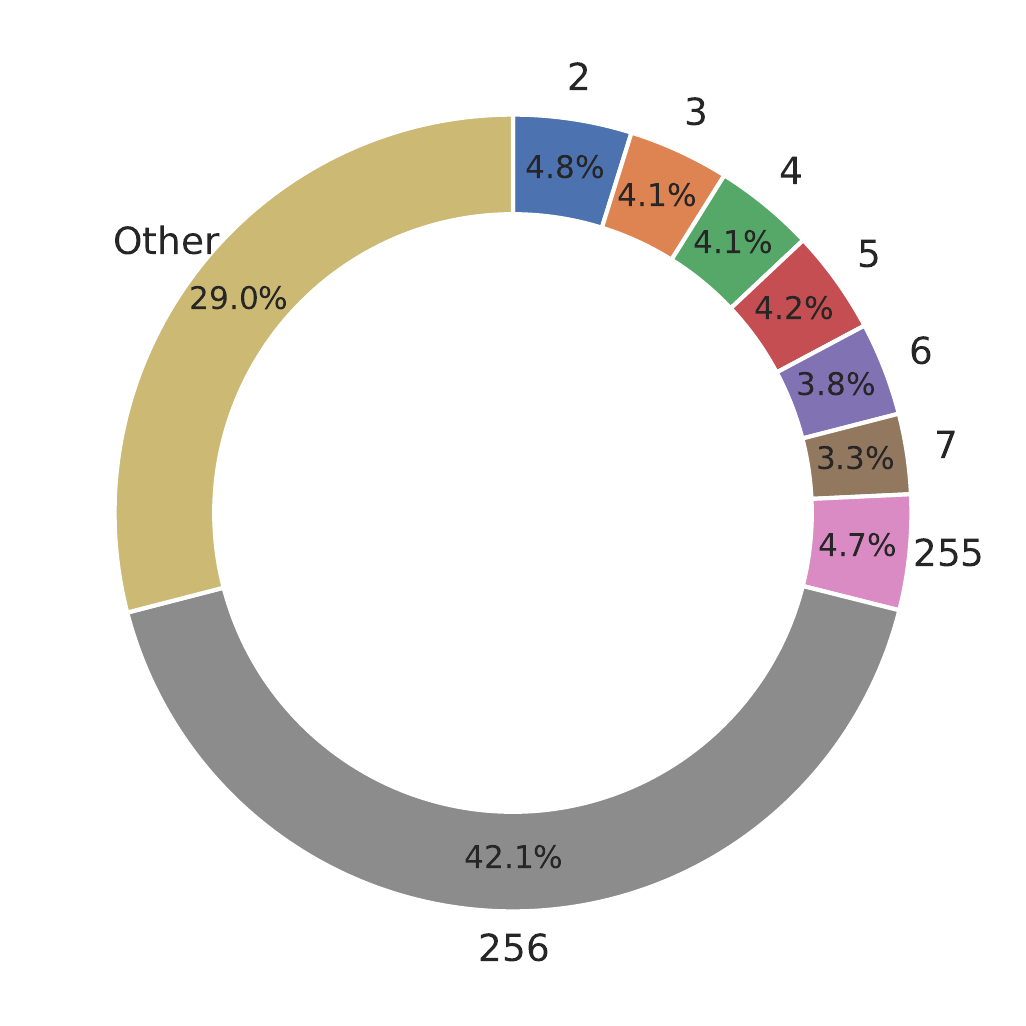}
        \caption{Inner}
        \label{fig:archive-inner-node-slot-usage}
    \end{subfigure}
    \hfill
    \begin{subfigure}{0.49\linewidth}
        \centering
        \includegraphics[width=\textwidth]{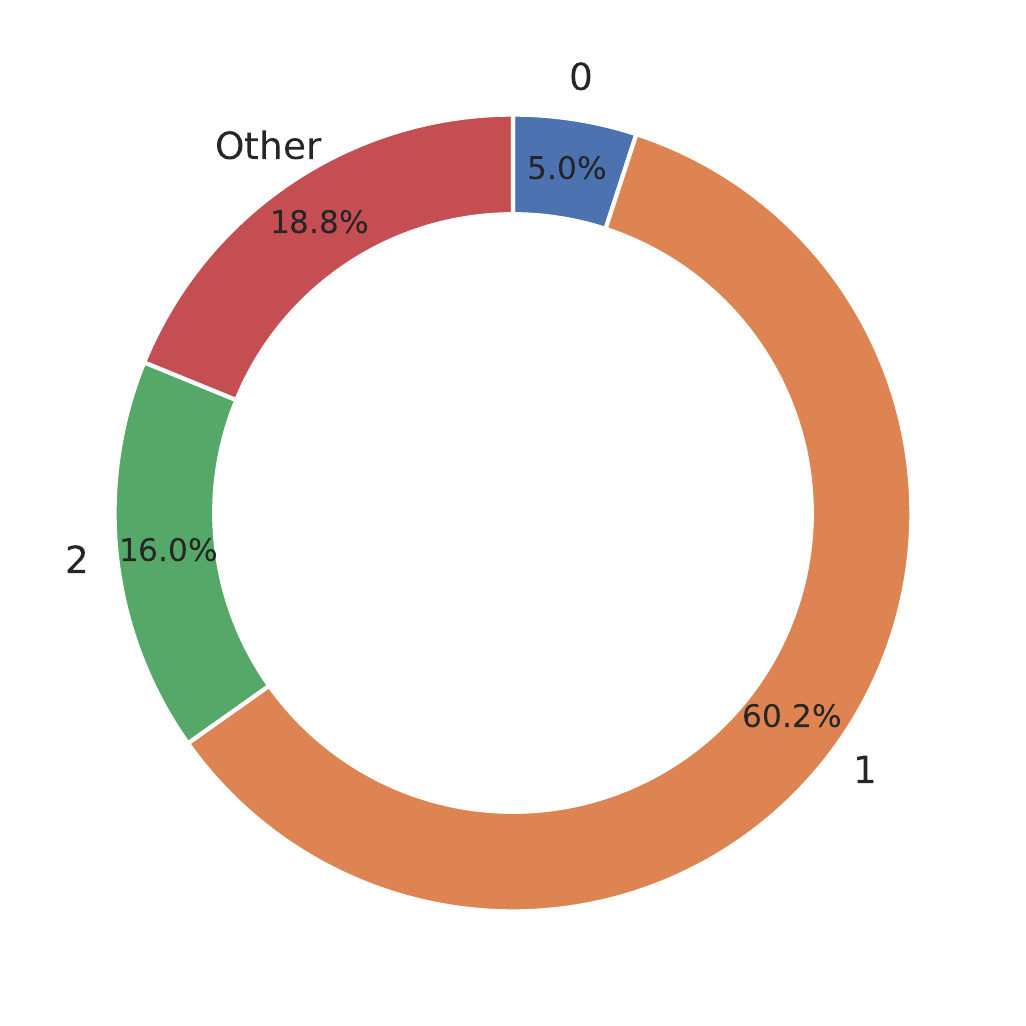}
        \caption{Leaf}
        \label{fig:archive-leaf-node-slot-usage}
    \end{subfigure}
    \caption{ArchiveDB: distribution of actual slot usage in Verkle nodes}
    \label{fig:archive-node-slot-usage}
\end{figure}

\begin{figure}[htbp]
    \centering
    \includegraphics[width=\textwidth]{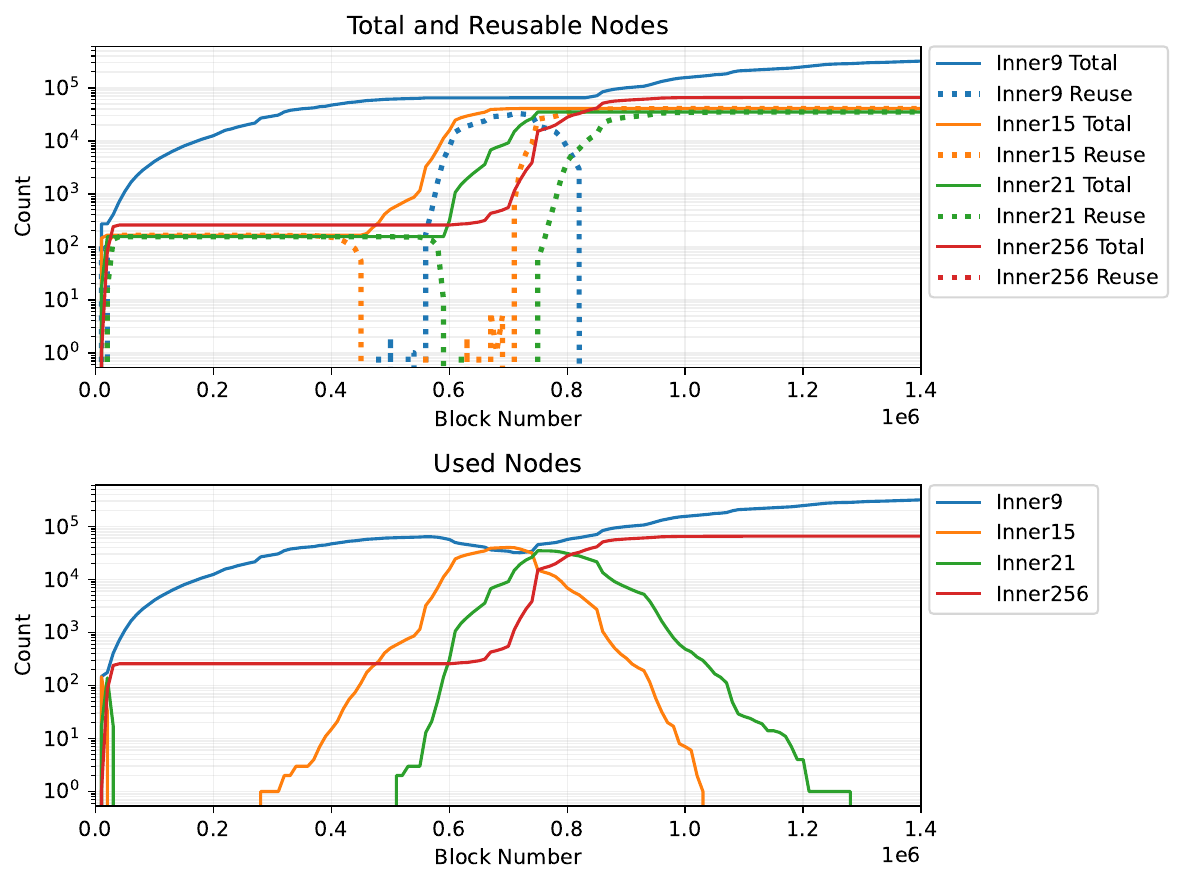}
    \caption{LiveDB: total, reusable, and used inner node counts}
    \label{fig:total-vs-reusable-node-counts}
\end{figure}

\begin{figure}[htbp]
    \centering
    \includegraphics[width=0.7\textwidth]{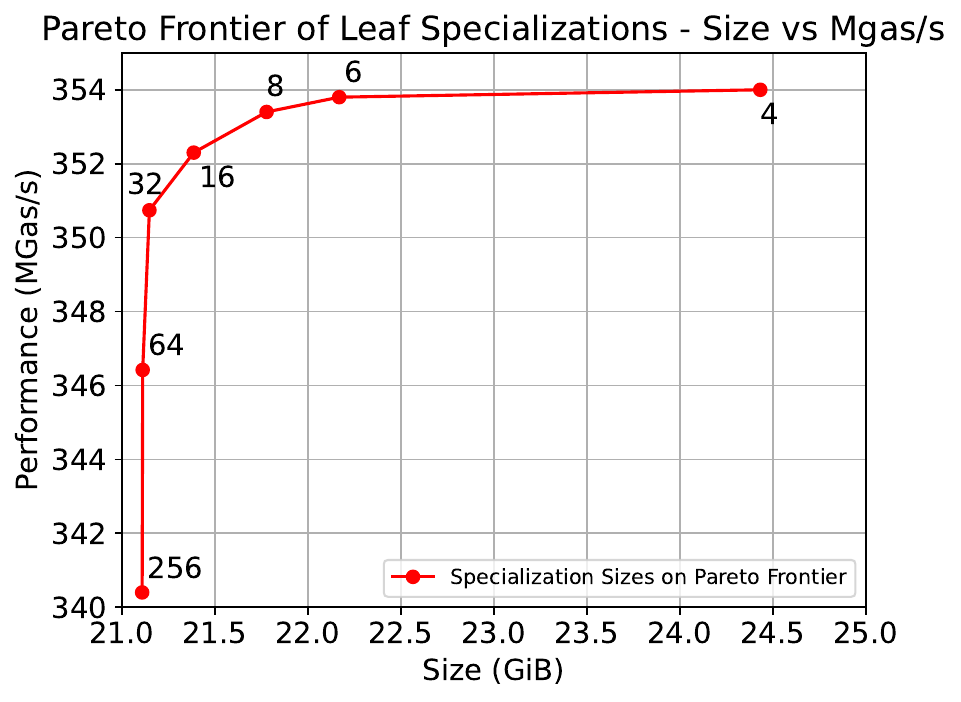}
    \caption{Storage-throughput trade-off for leaf-node specializations.
        Each point corresponds to the DP-optimal configuration for a given
        number of specializations~$k \in \{6, 8, 16, 32, 64, 128, 256\}$.  Throughput
        degrades monotonically with~$k$; storage savings beyond $k = 6$ are
        sub-percent. We select $k = 6$ as the smallest configuration at which
        additional specializations yield negligible storage improvement (see
        text).}
    \label{fig:pareto-frontier}
\end{figure}

Figure~\ref{fig:pareto-frontier} shows the measured trade-off between storage consumption and throughput as the number of leaf-node specializations~$k$ is varied, with each point corresponding to the configuration produced by the dynamic program of Section~\ref{sec:specialization} for that~$k$.
The figure isolates the throughput cost that the storage-only objective of the DP does not capture, thereby turning the choice of~$k$ into an empirical constrained optimization problem: minimize storage subject to a throughput floor imposed by Sonic's 300~ms block budget.

Two features of the trade-off curve drive the selection.
First, throughput degrades rapidly with~$k$ from~32 to~256.  We attribute this degradation to three effects that compound with specialization count: (i)~additional branch
targets in the node-access hot path, (ii)~a larger working set of type descriptors and free lists competing for L1/L2 cache, and (iii)~increased
free-list fragmentation as nodes migrate between specializations whose occupancy distributions shift over time (Figure~\ref{fig:total-vs-reusable-node-counts}).
Second, storage savings exhibit sharply diminishing returns beyond a small number of specializations.  Moving from $k = 4$ to $k = 32$ reduces storage from $24.5$ GiB to $ 21.2$~GiB, i.e., well under $1\%$ of the LiveDB footprint, while incurring several Mgas/s in throughput cost.

Under Sonic's operating constraints, throughput is the binding resource.
A validator that cannot close the 300~ms block window ceases to participate in consensus, whereas a validator with slightly higher storage consumption merely pays a marginal hardware cost.
We therefore select the smallest~$k$ at which additional specializations yield sub-percent storage improvements, which corresponds to $k = 6$ for leaf nodes.
Applying the same criterion to inner nodes yields $k = 4$.
We emphasize that this criterion is asymmetric by design: we accept a bounded storage overhead (a few hundred MiB) in exchange for protecting throughput headroom, and we would make the opposite choice on a workload where storage dominates the operating cost.

With~$k$ fixed, the dynamic program of Section~\ref{sec:specialization} selects the optimal specialization thresholds.
For~$k = 4$ inner-node specializations, the thresholds are $\{9, 15, 21, 256\}$ children; for~$k = 6$ leaf-node specializations, the thresholds are $\{1, 2, 5, 18, 146, 256\}$ values.
Both threshold sets reflect the heavy skew of the occupancy distribution: the first three or four specializations cover the overwhelming majority of nodes (Figures~\ref{fig:live-node-slot-usage} and~\ref{fig:archive-node-slot-usage}), and the final threshold at~$256$ is forced by the coverage constraint of Definition~\ref{def:specialization}.

\subsection{Throughput}

As a comparison, we used the Geth Verkle implementation.
However, the original Geth implementation of the Verkle trie runs entirely in memory; thus, the number of blocks that can be processed is limited by available memory.
Moreover, the verkle implementation has been discontinued and removed starting from Geth v1.17.0 \footnote{https://github.com/ethereum/go-ethereum/pull/33461}.
To enable a fair comparison, we extended the Geth Verkle implementation to include a LevelDB storage backend.
Our implementation \footnote{https://github.com/0xsoniclabs/carmen/tree/herbert/geth\_vt\_persistent/go/database/vt/geth2} depends directly on \texttt{go-verkle} \cite{go-verkle}, with some parts of the code such as the trie embedding ported from go-Ethereum v.16.9, and supports both Live and Archive modes.
The LevelDB backend follows the design of the go-ethereum MPT LevelDB backend closely to ensure a fair representation of the original intended implementation of the Verkle trie in Geth.
\newline

\begin{figure}[htbp]
    \centering
    \includegraphics[width=\textwidth]{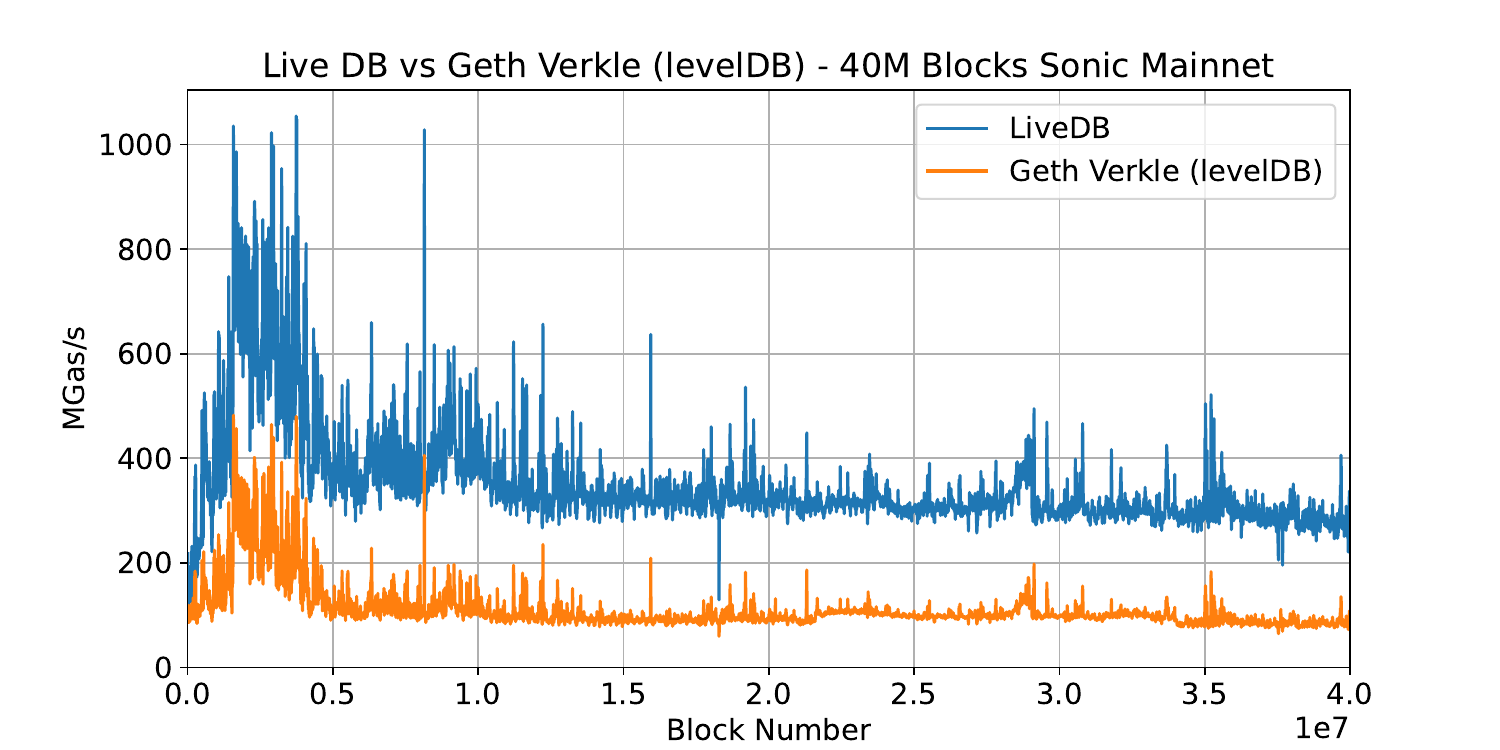}
    \caption{LiveDB vs Geth Verkle (LevelDB) performance for the first 40M blocks.}
    \label{fig:live-geth-leveldb-mgas}
\end{figure}

Figure \ref{fig:live-geth-leveldb-mgas} shows the performance comparison of the LiveDB and Geth Verkle (with LevelDB) for the first 40M blocks.
Our implementation outperforms the Geth Verkle implementation by a factor of 3.2x, taking 42h43m to process the first 40M blocks, compared to 137h53m for the Geth Verkle implementation.

\begin{figure}[htbp]
    \centering
    \includegraphics[width=\textwidth]{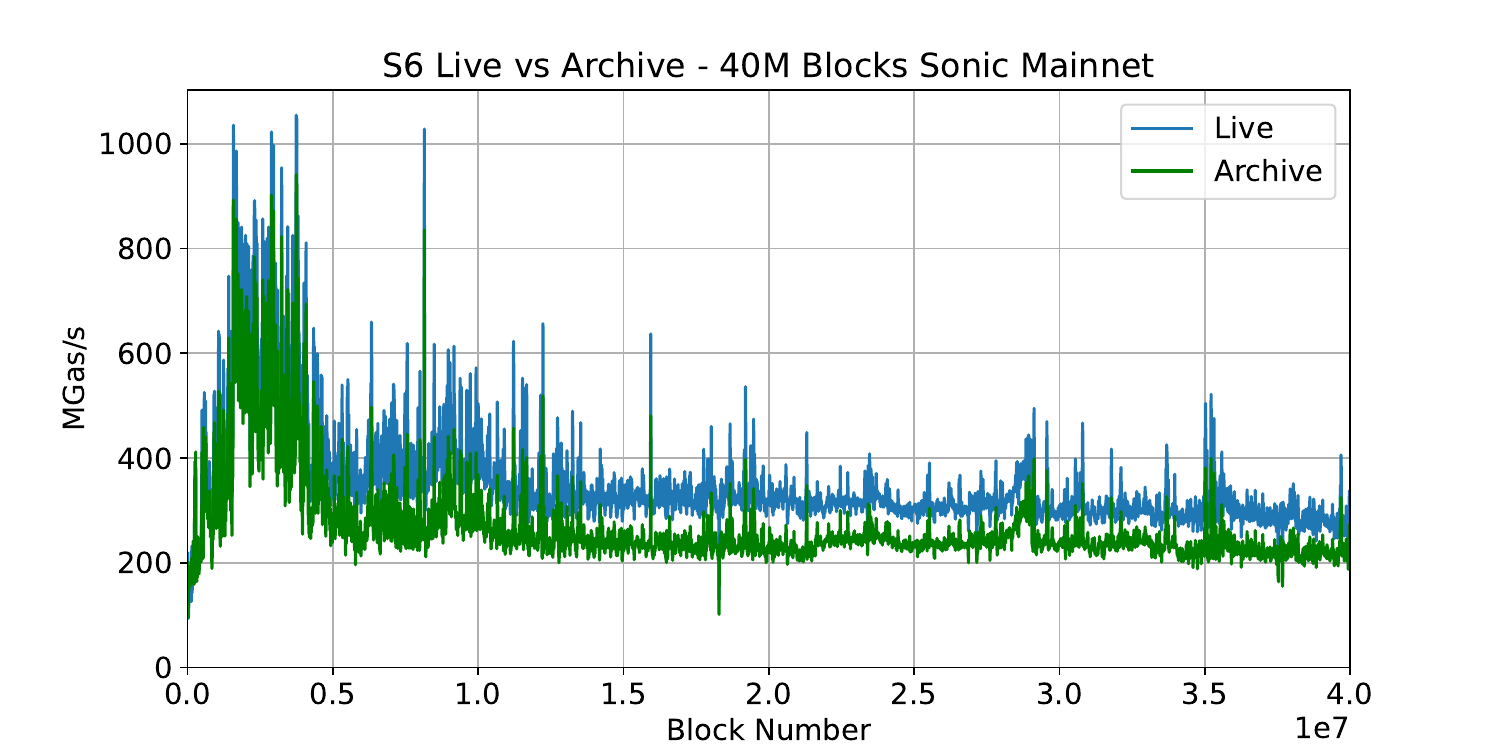}
    \caption{LiveDB vs ArchiveDb performance for the first 40M blocks.}
    \label{fig:live-archive-mgas}
\end{figure}

Figure \ref{fig:live-archive-mgas} shows the performance comparison of the LiveDB and ArchiveDb Rust implementations.
The archive implementation takes 56h12m, 24\% slower than the live implementation but still in the same order of magnitude, which is required for the archive to stay in sync with the network. \\

\begin{figure}
    \centering
    \includegraphics[width=\textwidth]{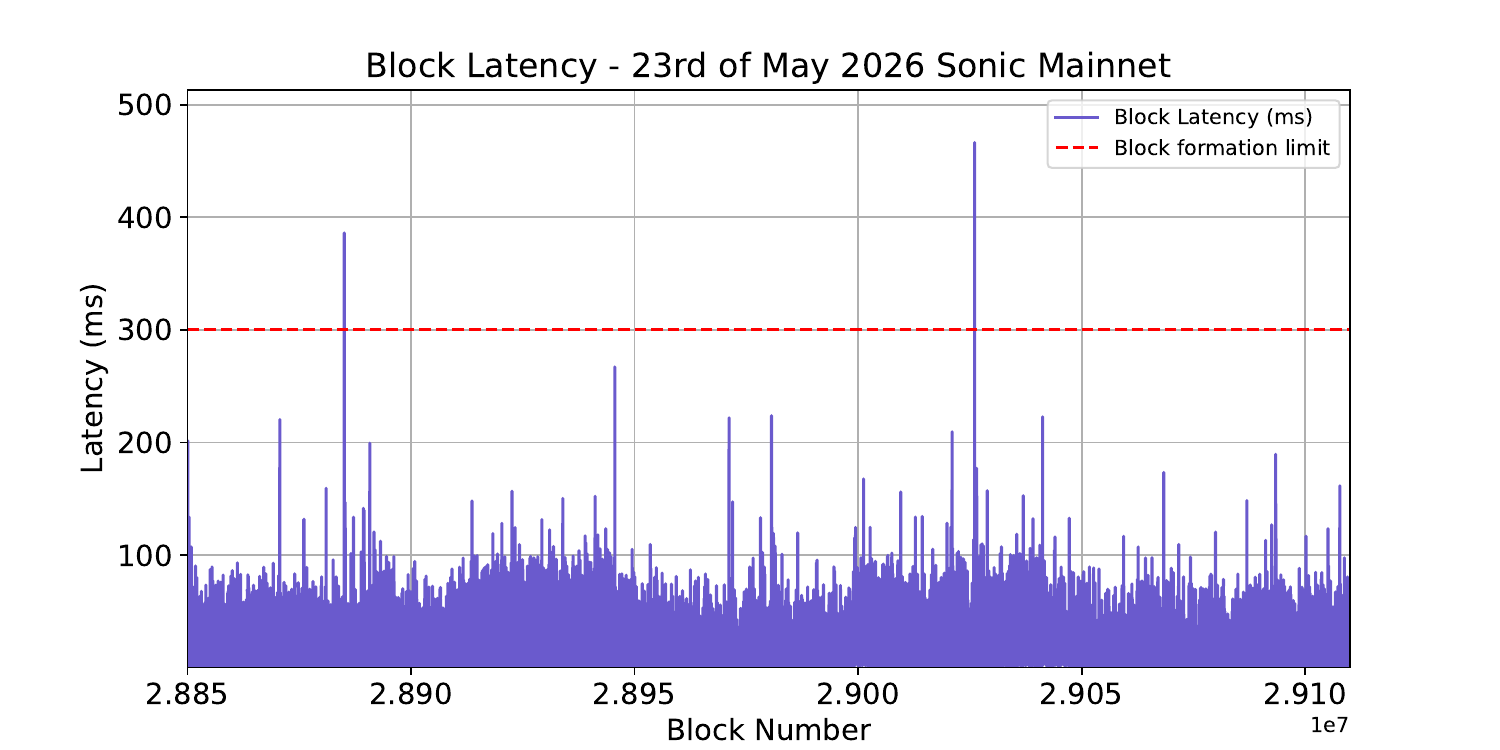}
    \caption{S6 LiveDB block latency - May 23, 2025 Sonic Mainnet}
    \label{fig:block-latency}
\end{figure}

Sonic requires blocks to be processed within a 300-millisecond window to keep up with the network's production rate.
The state database is a critical component of the block processing pipeline, and its performance directly impacts the block latency.
Figure \ref{fig:block-latency} shows the block processing latency of the LiveDB implementation for 24 hours on May 23, 2025, on the Sonic mainnet.
We chose this date because it represents the highest transaction volume in Sonic's history and therefore serves as a representative benchmark for a worst-case scenario for database performance.
The Verkle implementation operates within the required 300-millisecond block window for more than 99.99\% of the time, with only \textit{28885063}, \textit{29026048}, and \textit{29026050} exceeding the 300-millisecond block window with exceptional transaction load for the Sonic network, with up to 450ms block latency, which is still tolerable for the network.

\subsection{Disk Size}

\begin{figure}[htbp]
    \centering
    \includegraphics[width=\textwidth]{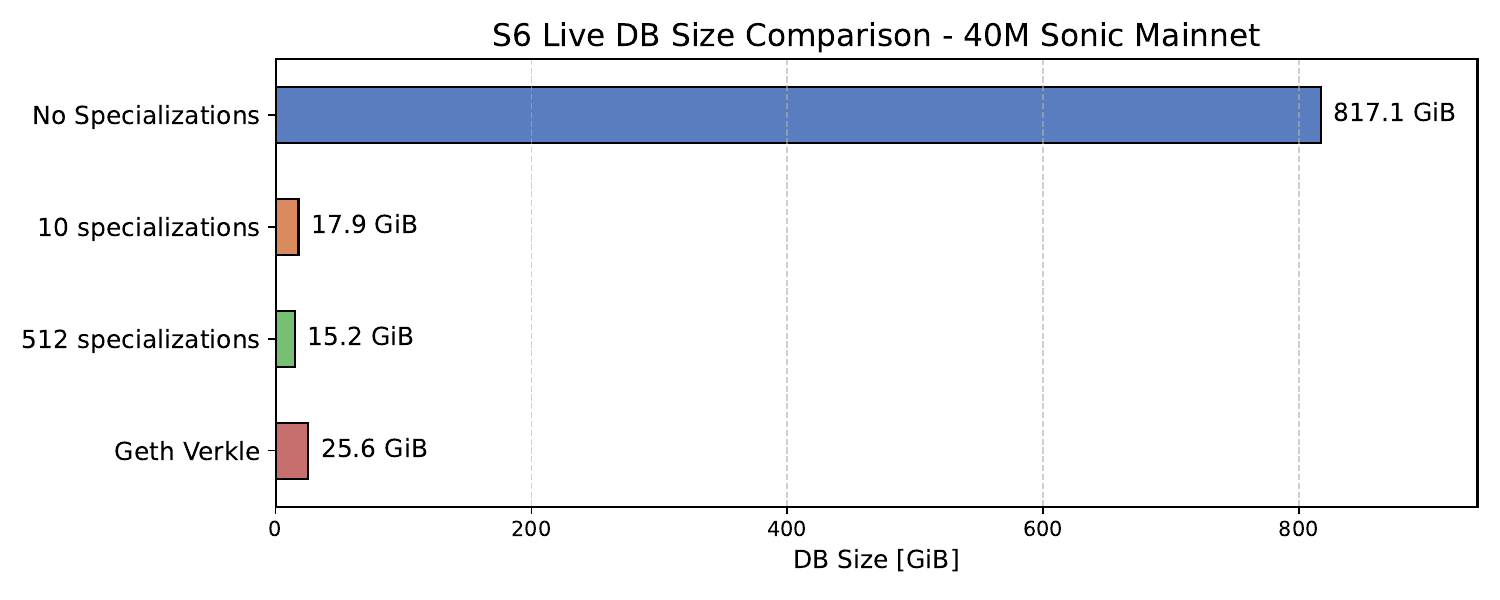}
    \caption{Comparison of LiveDb Sizes in GiB for 40M blocks}
    \label{fig:live-db-sizes}
\end{figure}

Figure \ref{fig:live-db-sizes} shows the comparison of LiveDB sizes for 40M blocks.
Using 10 specializations computed using the dynamic programming model reduces the storage space by 97.8\%, from 817.1 GiB to 17.9 GiB.
Using all 512 possible specializations reduces storage space by an additional 15\%, achieving 15.2 GiB, but introduces a significant performance overhead, as shown in Figure \ref{fig:pareto-frontier}.
On the other hand, Geth Verkle uses 25.6 GiB of storage, which is 40\% larger than our Rust file implementation.
It is worth noting that the Geth Verkle LevelDB backend compresses the stored data, while our file implementation does not.

\begin{figure}[htbp]
    \centering
    \includegraphics[width=\textwidth]{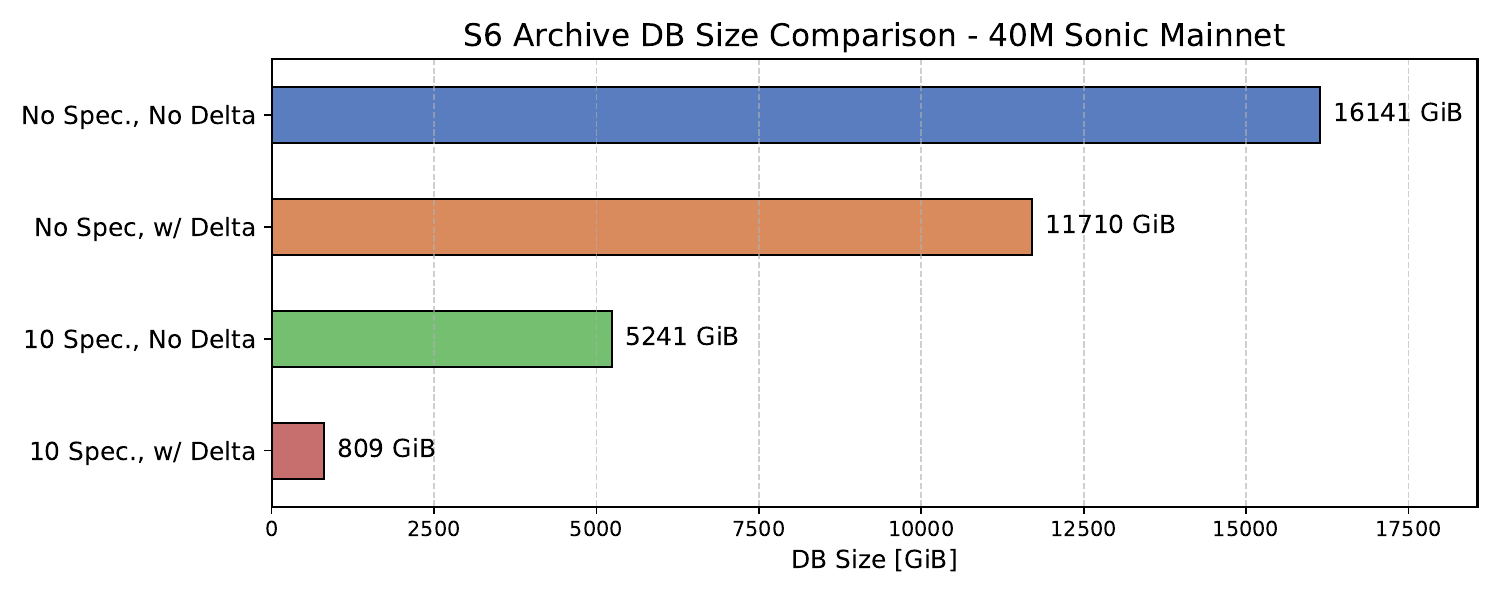}
    \caption{Comparison of Archive Sizes in GiB for 40M blocks}
    \label{fig:archive-db-sizes}
\end{figure}

Figure \ref{fig:archive-db-sizes} shows a comparison of ArchiveDB sizes for 40M blocks.
Here, we have used two combinations of optimizations: node specializations and delta nodes.
We can see that the two optimizations are complementary and both contribute to the overall reduction in size.
The delta node alone reduces the size by 25.48\%, while the node specializations alone reduce the size by 67.53\%. The combination of the two optimizations reduces the size by 94.99\%, going from 16141 GiB to 809 GiB.

\begin{figure}[htbp]
    \centering
    \includegraphics[width=\textwidth]{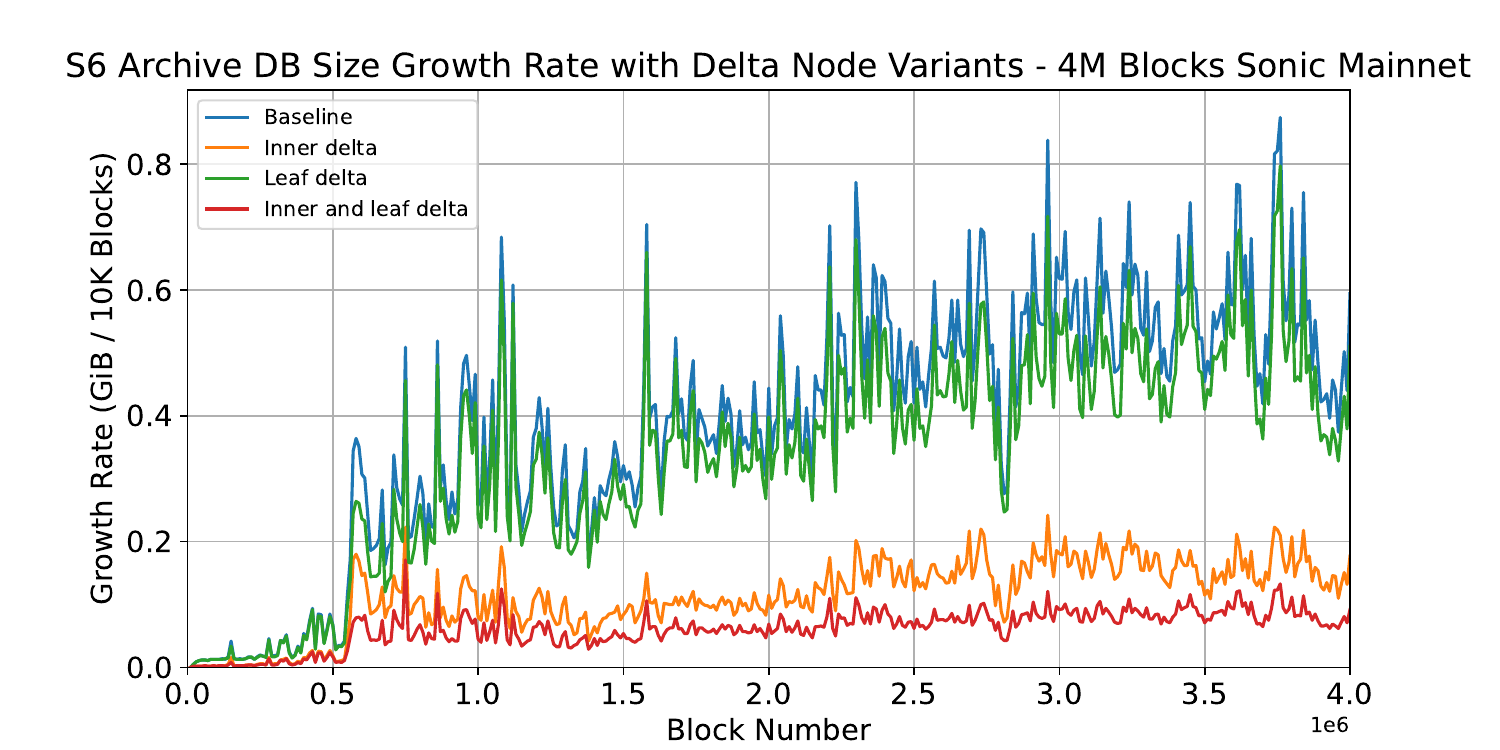}
    \caption{Archive DB growth rate with different combinations of delta node implementations.}
    \label{fig:archive-db-growth-rate-delta-comparison}
\end{figure}

As shown before, the delta nodes are responsible for a significant reduction in the overall size of the archive database.
To better understand the impact of delta nodes on the archive database's growth rate, we can examine how different combinations of delta node implementations affect it.
Figure \ref{fig:archive-db-growth-rate-delta-comparison} shows the growth rate of the archive database for the first 4M blocks with different combinations of delta node implementations.
The inner delta nodes have the greatest impact on the archive database's growth rate, as they grow quickly due to ArchiveDB's copy-on-write nature.

\section{Related Work\label{sec:related}}

\paragraph{Authenticated storage for blockchain state.} Raju et al. \cite{raju2018mlsm} propose to use merkelized log-structured merge trees (mLSM) as the storage backend for Ethereum.
Their main objectives are to reduce read and write amplification while still providing authenticated storage.
To this end, they build on LSM trees, using immutable binary Merkle Tries instead of SSTables for authenticated storage.
The root hashes of all sub-trees are combined into a single root hash.
They maintain a lookup cache for each level of the LSM tree to speed up reads, and a bloom filter for each sub-tree to quickly determine whether a key may be present in that sub-tree.
Furthermore, they batch updates in memory and write them as a single contiguous chunk.

Reads of frequently used keys are served from the cache.
For uncached keys, bloom filters help quickly determine whether a key may be in a sub-tree and avoid unnecessary tree traversals.
This means a cache hit exhibits a worst-case runtime complexity of $\mathcal{O}(1)$ disk operations.
A cache miss has an amortized complexity of $\mathcal{O}(\log N)$, where $N$ is the number of nodes in the sub-tree containing the key (amortized because Bloom filters may produce false positives, which would require multiple tree traversals).
And because there are multiple sub-trees, each containing only a fraction of the total keys, the depth of each sub-tree is smaller than that of a single large tree containing all keys.

Insertions and updates are batched and written out as a new Merkle Trie.
This allows for one contiguous write operation, which is more efficient than many random writes.
Creating a new sub-tree may trigger compactions of existing sub-trees.
However, each update is guaranteed to be written only once to each level of the LSM tree.
This means in the worst case, there are $\mathcal{O}(H)$ writes for each update, where $H$ is the number of levels in the LSM tree.

One challenge with authenticated storage is that the LSM tree's compaction algorithm must be deterministic.
However, this is only an engineering challenge.

Choi et al. proposed the Layered Merkle Patricia Trie (LMPT) \cite{LMPT}, a high-performance data structure for transaction processing systems.
The main improvement over MPT is to split the trie into three smaller tries: a flat key-value store, a snapshot MPT serialized on disk, and an intermediate and a delta MPT in memory, serving as L2 and L1 caching layers, respectively.
Reading operations follow a hierarchical organization, first in the Delta MPT, then in the Intermediate MPT, and finally in the storage, thereby significantly reducing I/O amplification.
LMPT maintains both a flat key-value and an MPT view over the stored data, enabling fast data retrieval for non-authenticated reads with the former and authenticated reads with the latter.
Writes are stored on the delta MPT, while a background thread merges the intermediate MPT into the underlying storage, allowing non-blocking updates.
Upon reaching a specified threshold, the intermediate MPT is flushed to storage and replaced with the actual delta MPT, leaving the higher cache level empty.
Compared with our approach, LMPT is significantly more difficult to implement because it requires maintaining and synchronizing three MPTs.

Yang et al. \cite{yang2025salt} propose SALT, which reduces the read and write amplification incurred by authenticated reads when using a pure tree data structure and improves the storage efficiency for the lower sparse levels of the tree.
They use a combination of a complete tree for the upper levels and flat buckets for the lower levels.
Each bucket represents one sub-tree, but flattens all nodes into a strongly history-independent hash table.
This way, the sparsity is no longer a problem.
Additionally, there is now only a single commitment for each bucket, in addition to the commitments for the nodes in the upper part.
The upper part of the tree and the commitments of the buckets are small enough to reside in memory.
This means that authentication does not require any disk accesses.
Since the upper part of the tree is in memory and hash maps allow point queries, both read and write operations require only a single disk access (except for bucket resizing).

NOMT (Nearly-Optimal Merkle Trie) \cite{NOMT} is a merkelized database optimized for commodity hardware.
In particular, it is designed to exploit modern SSD performance by packing data into 4KB pages to maximize SSD bandwidth and employing several techniques to hide latency, e.g., batching and DMA.
NOMT is composed of two main components: Bitbox, a Merkle page store used for authenticated reads, and Beatree, a raw key-value store implemented as a custom B-Tree for serving state reads in 1 I/O.
Each page stores a rootless B-tree of depth 6, allowing retrieval of $2^6$ nodes in a single I/O.
In addition, NOMT exploits DMA CPU idle time to generate witness proofs during fetch operations.
However, the high latency of SSDs necessitates batching queries, a latency-hiding technique that requires a parallel EVM.
In addition, it is unclear how an archive could be implemented, as it would require page duplication to keep reads aligned, incurring up to 4KB of 32-byte memory overhead for a single-node update.

LVMT~\cite{lvmt} uses authenticated multi-point evaluation trees (AMT)~\cite{tomescu2020towards} as the underlying commitment scheme for an authenticated storage data structure.
AMTs are an extension of KZG polynomial commitments~\cite{kate2010constant} that offer faster proof generation at the cost of a logarithmic (instead of constant) proof size and the need to maintain additional proof-generation metadata.
By storing \textit{key-version} pairs within a hierarchy of AMTs, LVMT avoids costly elliptic curve multiplications on updates, as each value update increments the version by one.
Values are stored as \textit{key-version-value} triples in a separate append-only MPT that is created for each block.
While LVMT has been shown to boost blockchain system throughput on real Ethereum traces by up to 2.7x, its multi-tree architecture introduces significant complexity for both implementation and proof verification, and the large amounts of required metadata necessitate sharding techniques to achieve good performance.

\paragraph{Delta Nodes.}
The idea of representing a versioned update as a small set of modifications against a base version, with a promotion step once a capacity threshold is exceeded, originates in the persistent
data-structure literature.
Driscoll, Sarnak, Sleator, and Tarjan~\cite{driscoll1989persistent} introduced the \emph{fat-node method with modification boxes}, where each node carries a fixed-capacity log of timestamped modifications and is split into a fresh copy once the log fills.
Our delta nodes are a disk-resident, coarse-grained adaptation: the threshold~$\tau$ serves as the fat-node capacity, and the base node serves as the node being modified.
The adaptation targets a different regime — disk-resident storage, block-height-indexed queries, a non-forking version sequence, and a hard per-query I/O budget — which motivates the design choices
above.
The Bw-tree~\cite{levandoski2013bwtree} applies a related delta-based strategy to a latch-free B-tree: page updates are prepended to a delta chain and periodically consolidated into a new base page.
Unlike SonicDB, the Bw-tree permits arbitrarily long chains between consolidations and bounds their length only in expectation; the difference reflects its target (in-memory OLTP latency) rather than on-disk archive footprint.
LSM trees~\cite{oneil1996lsm} buffer updates at internal nodes and flush them down in batches, a pattern architecturally distinct from base-plus-delta but motivated by the same desire to avoid full-node rewrites on small updates.

\section{Conclusion\label{sec:conclusion}}

We presented SonicDB S6, a production Rust implementation of a Verkle trie state database for the Sonic blockchain.
Motivated by Sonic's 300-millisecond block time and the need to serve both live and historical state queries in real time, we developed three principal contributions.
First, occupancy-aware node specializations with a subsumption property, optimally selected via an $\mathcal{O}(kn^2)$ dynamic programming algorithm, reduce live database storage by 97.8\%, from 817.1 GiB to 17.9 GiB across 40 million blocks of Sonic history.
Second, delta nodes that store only changed slots relative to a base node, without chaining, reduce archive storage by 95\%, from 16,141 GiB to 809 GiB, making it feasible to operate a full archive node on commodity hardware. Third, a layered architecture that combines batched trie updates, multi-threaded commitment computation via a work-stealing thread pool, and homomorphic caching of intermediate Pedersen commitment state achieves $3.2\times$ higher throughput than a persistent Geth Verkle baseline while maintaining archive performance within production block-rate requirements.

\section*{Acknowledgment}
We would like to thank Philip Salzmann for the implementation contributions for S6 and the Sonic Labs team. 

\bibliographystyle{plain}
\bibliography{main}

\end{document}